\documentclass[11pt]{article}
\usepackage[a4paper,margin=1in]{geometry}
\usepackage[T1]{fontenc}
\usepackage[utf8]{inputenc}
\usepackage{lmodern}
\usepackage{microtype}
\usepackage{amsmath,amssymb,amsthm,mathtools}
\usepackage{graphicx}
\usepackage{float}
\usepackage[font=small,labelfont=bf]{caption}
\usepackage{xcolor}
\usepackage[numbers,sort&compress]{natbib}
\usepackage[colorlinks=true,linkcolor=blue!60!black,citecolor=blue!60!black,urlcolor=blue!70!black]{hyperref}
\usepackage{doi}

\newtheorem{theorem}{Theorem}

\newtheorem{corollary}[theorem]{Corollary}
\newtheorem{lemma}[theorem]{Lemma}
\theoremstyle{remark}

\DeclareMathOperator{\Tr}{Tr}
\renewcommand{\Re}{\operatorname{Re}}
\renewcommand{\Im}{\operatorname{Im}}
\newcommand{\cD}{\mathcal D}
\newcommand{\cL}{\mathcal L}

\graphicspath{{figures/}}

\title{\textbf{Boundary R\'enyi oscillations in the open XX chain}}
\author{Miguel Tierz\\[4pt]
\small Shanghai Institute for Mathematics and Interdisciplinary Sciences (SIMIS),\\
\small Block A, International Innovation Plaza, No.\ 657 Songhu Road,\\
\small Yangpu District, Shanghai, China\\[2pt]
\small \texttt{tierz@simis.cn}}
\date{}

\hypersetup{
  pdftitle={Boundary R\'enyi oscillations in the open XX chain},
  pdfauthor={Miguel Tierz}
}

\begin{document}
\maketitle

\begin{abstract}
For the ground state of the semi-infinite XX chain with Fermi momentum $k_F$, we derive the extracted $2k_F$ contribution to the R\'enyi entropy of its first $\ell$ sites at fixed $\alpha>1$. Its phase $\sin[(2\ell+1)k_F]$ matches the Friedel oscillation at the entangling bond and becomes the even--odd alternation at half filling. An exact Hankel determinant reduction and Riemann--Hilbert analysis give every term $\cL^{-(2q-1)/\alpha}$ for positive integers $q$ with $2q-1<\alpha$, where $\cL=4\ell\sin k_F$, with an $O(\cL^{-1})$ error within this extracted component for $k_F$ away from $0$ and $\pi$. At integer $\alpha=p\ge2$, exact factorization gives a finite sum over all roots with error $O(\cL^{-1-1/p})$. We also determine the $2k_F$ terms of $\log Z_\alpha(\phi)$, where $Z_\alpha(\phi)=\Tr(\rho_A^\alpha e^{i\phi Q_A})$, $\rho_A$ is the block's reduced density matrix, and $Q_A$ is its fermion number. At fixed nonzero flux $|\phi|<\pi$ with $2q-1\mp\phi/\pi\ne\alpha$, the two exponential terms have different decay powers. Exact factorization sharpens the integer result. These pointwise estimates do not justify projection to fixed charge. The uncharged sequence and the two leading charged terms agree with earlier generalized Fisher--Hartwig predictions. Numerical calculations at $p=4$ test the predicted subleading coefficients without an adjustable normalization.
\end{abstract}

\section{Introduction and main results}
\label{sec:introduction}

We consider a semi-infinite XX chain with an open end. The subsystem $A=[1,\ell]$ consists of its first $\ell$ sites, and $\bar A=[\ell+1,\infty)$ is the rest of the chain. The entangling cut lies between sites $\ell$ and $\ell+1$, as shown in Fig.~\ref{fig:interval-cut}. For the ground state $|\Omega\rangle$, tracing over $\bar A$ gives $\rho_A=\Tr_{\bar A}|\Omega\rangle\langle\Omega|$.

The R\'enyi entropy of $A$ is
\begin{equation}
 S_\alpha(A)=\frac{1}{1-\alpha}\log\Tr\rho_A^\alpha,
 \qquad \alpha>0,\quad \alpha\ne1.
 \label{eq:renyi-definition}
\end{equation}
The von Neumann entropy $S_1(A)=-\Tr(\rho_A\log\rho_A)$ is the continuous limit $\alpha\to1$.
For an interval of length $\ell$ attached to the boundary of a critical one-dimensional system, boundary conformal field theory gives the smooth logarithmic term
\begin{equation}
 S_\alpha(\ell)=
 \frac{c}{12}\left(1+\frac{1}{\alpha}\right)\log\frac{\ell}{a_{\rm UV}}
 +C_\alpha+\cdots,
 \label{eq:bcft-log}
\end{equation}
where $c$ is the central charge, $a_{\rm UV}$ is a short-distance cutoff, and $C_\alpha$ contains the boundary contribution~\cite{HolzheyLarsenWilczek1994,CalabreseCardy2004JSM,CalabreseCardy2009JPA}, including the Affleck--Ludwig boundary entropy~\cite{AffleckLudwig1991}. The open XX chain has $c=1$. This conformal logarithm does not contain the leading correction from the lattice.

\begin{figure}[H]
 \centering
 \includegraphics[width=0.84\linewidth]{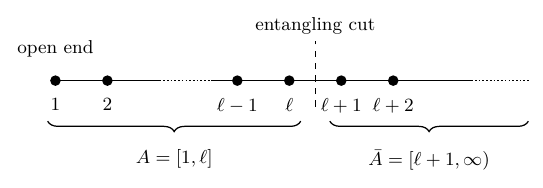}
 \caption{A block attached to the open end of the chain. The dashed line marks the bipartition used to define $\rho_A$. The hopping bond between $\ell$ and $\ell+1$ remains present.}
 \label{fig:interval-cut}
\end{figure}

The open boundary reflects the fermions and produces a standing-wave density pattern. Let $c_j$ annihilate a Jordan--Wigner fermion at site $j$, and let
$C_{jk}=\langle c_j^\dagger c_k\rangle$ denote the full ground-state one-body correlator. The Fermi momentum is $k_F\in(0,\pi)$, and the bulk fermion density is $k_F/\pi$. The correlator at the bond crossed by the entangling cut in Fig.~\ref{fig:interval-cut} is
\begin{equation}
 C_{\ell,\ell+1}
 =\frac{\sin k_F}{\pi}
 -\frac{\sin[(2\ell+1)k_F]}{\pi(2\ell+1)}.
 \label{eq:cut-bond}
\end{equation}
The second term is the $2k_F$ Friedel contribution evaluated at this bond. For fixed $\alpha>1$, the leading term in the extracted entropy component is
\[
 \delta S_{\alpha,q=1}^{(1)}(\ell)=
 \frac{2}{1-\alpha}
 \frac{\Gamma(\frac12+\frac{1}{2\alpha})}
      {\Gamma(\frac12-\frac{1}{2\alpha})}
 (4\ell\sin k_F)^{-1/\alpha}\sin[(2\ell+1)k_F].
\]
Thus the entropy correction has the same phase as the bond oscillation. At half filling, $k_F=\pi/2$, the factor $\sin[(2\ell+1)k_F]$ becomes $(-1)^\ell$. The term ``first harmonic'' will mean this spatial $2k_F$ component. Higher spatial harmonics have wave numbers $4k_F,6k_F,\ldots$~\cite{Laflorencie2006PRL,CalabreseCampostriniEsslerNienhuis2010PRL,AffleckLaflorencieSorensen2009}.

The Fisher--Hartwig approach to the entanglement determinant of the XX chain was introduced by Jin and Korepin~\cite{JinKorepin2004}, and Calabrese and Essler obtained the periodic-chain oscillatory corrections from a generalized Fisher--Hartwig expansion~\cite{CalabreseEssler2010JSM}. Fagotti and Calabrese obtained the open-chain oscillatory expansion from a proposed generalized Fisher--Hartwig representation of the determinant associated with the restricted correlation matrix~\cite{FagottiCalabrese2010JSM}. Their expansion is indexed first by the spatial harmonic and then by the residues contributing to that harmonic. Its $2k_F$ component contains the sequence $(4\ell\sin k_F)^{-(2q-1)/\alpha}$, $q=1,2,\ldots$. Our purpose is to derive directly from the determinant problem all terms in this component whose decay power is below one and to bound its remainder.

The proof has three steps. Section~\ref{sec:determinant-reduction} uses the exact identity \eqref{eq:even-symbol-identity} to convert the Toeplitz-plus-Hankel correlation determinant into a scalar Hankel determinant on $[-1,1]$ with one internal jump at $x=\cos k_F$~\cite{DeiftItsKrasovsky2011,BasorEhrhardt2001,BasorEhrhardt2002}. Its jump parameter $\beta(\lambda)$, defined in Eq.~\eqref{eq:beta-definition}, is generally complex on the entropy contour. Section~\ref{sec:complex-jump} starts from the positive-jump Riemann--Hilbert construction and recurrence asymptotics of Refs.~\cite{KuijlaarsRH2004,FoulquieMorenoMartinezFinkelshteinSousa2011,FoulquieMorenoMartinezFinkelshteinSousaJAT2010}, together with the leading complex-parameter analysis of Ref.~\cite{CharlierGharakhloo2021}, and obtains the two oscillatory determinant coefficients. The exact second-difference identity \eqref{eq:second-difference-summary} converts the recurrence result into the determinant result. Its action on the two oscillatory factors is displayed in Eq.~\eqref{eq:sector-multiplier}.

Section~\ref{sec:renyi-oscillations} inserts those determinant coefficients into the entropy contour integral after integration by parts. Deformation on the initial sheet gives the residues at the zeros of $F_\alpha$ and an integral along its two real branch cuts, as stated in Eq.~\eqref{eq:residue-decomposition}. Appendix~\ref{app:positive-jump-uniform} supplies the uniform positive-jump estimate needed to bound that integral. At integer index, the exact factorization \eqref{eq:integer-factorization-main} avoids the contour deformation.

For fixed $\alpha>1$, Theorem~\ref{thm:real-alpha} evaluates every residue contribution to the extracted $2k_F$ component whose decay power is below one. The complete coefficient formula and its $O(\cL^{-1})$ remainder are stated together in Eq.~\eqref{eq:main-real-alpha}. The retained powers are $\cL^{-(2q-1)/\alpha}$ with $2q-1<\alpha$, where $\cL=4\ell\sin k_F$. At integer index $p\ge2$, Corollary~\ref{cor:integer-p} gives the finite-root expansion \eqref{eq:main-integer-p} with the sharper remainder $O(\cL^{-1-1/p})$. Its $p=4$ specialization is Eq.~\eqref{eq:p4-example}. Both estimates concern only the extracted $2k_F$ component, not the full entropy remainder.

Section~\ref{sec:charged} treats the charged moment defined in Eq.~\eqref{eq:charged-moment-definition}. At nonzero counting flux, the two determinant phases acquire different decay powers. Theorem~\ref{thm:charged} gives the fixed-flux expansion in Eq.~\eqref{eq:main-charged}, including every finite zero on the initial sheet whose decay power is below one. At integer index, Corollary~\ref{cor:charged-integer} gives the sharpened expansion \eqref{eq:main-charged-integer}. The $p=4$, $\phi=\pi/2$ case is Eq.~\eqref{eq:p4-charged-example}. These pointwise fixed-flux estimates do not supply the uniform control near the endpoints needed for projection to fixed charge.

The two leading charged terms agree with Bonsignori and Calabrese~\cite{BonsignoriCalabrese2021}. For a bulk interval in an infinite translation-invariant chain, Fraenkel and Goldstein obtained the analogous flux-dependent sequence, with doubled bulk decay powers, from a generalized Fisher--Hartwig expansion~\cite{FraenkelGoldstein2020JSM}. For a block attached to the boundary, we obtain the remaining residue terms in the $2k_F$ component and bound their remainder.

We consider a block attached to the boundary, with fixed $\alpha>1$ and $k_F$ bounded away from $0$ and $\pi$. The estimates are not uniform at the filling endpoints, where the internal jump approaches a Jacobi endpoint and the multiplier in Eq.~\eqref{eq:sector-multiplier} vanishes. Intervals separated from the boundary, several disjoint blocks~\cite{AresEsteveFalceto2014}, and the endpoint limits are not treated. Section~\ref{sec:numerics} tests the predicted subleading terms, and the appendices give the Riemann--Hilbert, contour, and factorization details.

The Hankel determinant approach and the boundary oscillations are discussed in the recorded talk~\cite{TierzSIMIS2026}.

\section{XX chain and determinant reduction}
\label{sec:determinant-reduction}

After a Jordan--Wigner transformation, the semi-infinite open XX chain is a hopping model of spinless fermions~\cite{PeschelEisler2009},
\begin{equation}
 H=-\frac12\sum_{j\ge1}
 \left(c_j^\dagger c_{j+1}+c_{j+1}^\dagger c_j\right)
 -h\sum_{j\ge1}c_j^\dagger c_j.
 \label{eq:Hamiltonian}
\end{equation}
Here $c_j$ annihilates a fermion at site $j$, and $h\in(-1,1)$ is the chemical potential. The Fermi momentum is $k_F=\arccos(-h)\in(0,\pi)$. In the ground state, the full one-body correlator is~\cite{PeschelEisler2009,FagottiCalabrese2010JSM}
\begin{equation}
 C_{jk}=\frac{\sin[k_F(j-k)]}{\pi(j-k)}
 -\frac{\sin[k_F(j+k)]}{\pi(j+k)},
 \qquad j,k\ge1,
 \label{eq:correlation-matrix}
\end{equation}
with the diagonal understood by continuity. Its restriction to $A=[1,\ell]$ is the $\ell\times\ell$ correlation matrix used below, and we denote this restriction again by $C$. The first term depends on $j-k$ and is Toeplitz. The image term depends on $j+k$ and is Hankel.

Let $\nu_1,\ldots,\nu_\ell$ be the eigenvalues of $C$. For $0<k_F<\pi$, they satisfy $0<\nu_j<1$, so $x_j=2\nu_j-1$ lies in $(-1,1)$. Then
\begin{equation}
 S_\alpha(\ell)=\frac{1}{1-\alpha}
 \sum_{j=1}^{\ell}\log\left[\nu_j^\alpha+(1-\nu_j)^\alpha\right].
 \label{eq:entropy-spectrum}
\end{equation}
For the complex spectral parameter $\lambda$, define the correlation determinant
\begin{equation}
 D_\ell(\lambda)=\det\left[(\lambda+1)I_\ell-2C\right]
 =\prod_{j=1}^{\ell}(\lambda-x_j).
\label{eq:physical-determinant}
\end{equation}
Here $I_\ell$ is the $\ell\times\ell$ identity matrix.
Choose a counterclockwise contour $\mathcal C$ enclosing $0,x_1,\ldots,x_\ell$ in a sufficiently small neighborhood of $(-1,1)$, excluding the branch points $\pm1$ and the zeros of $\bigl(\tfrac{1+\lambda}{2}\bigr)^{\alpha}+\bigl(\tfrac{1-\lambda}{2}\bigr)^{\alpha}$, at which $e_\alpha$ has logarithmic singularities. These are the points $\lambda_{q,\pm}$ of Eq.~\eqref{eq:kernel-zeros}. The entropy has the standard exact spectral-contour representation~\cite{FagottiCalabrese2010JSM}
\begin{equation}
 S_\alpha(\ell)=\frac{1}{2\pi i}\oint_{\mathcal C}
 e_\alpha(\lambda)\frac{d}{d\lambda}\log D_\ell(\lambda)\,d\lambda,
 \label{eq:entropy-contour}
\end{equation}
where
\begin{equation}
 e_\alpha(\lambda)=\frac{1}{1-\alpha}
 \log\left[
 \left(\frac{1+\lambda}{2}\right)^\alpha+
 \left(\frac{1-\lambda}{2}\right)^\alpha
 \right].
 \label{eq:renyi-kernel}
\end{equation}
On the initial sheet, each power is defined by
$z^\alpha=\exp(\alpha\operatorname{Log}z)$ with the principal logarithm. Thus the cuts of $(1+\lambda)^\alpha$ and $(1-\lambda)^\alpha$ are $(-\infty,-1]$ and $[1,\infty)$, respectively. The contour and its deformation remain on this sheet. To integrate by parts, normalize the determinant by $\lambda^\ell$. Its exterior logarithm vanishes at infinity and is single-valued along $\mathcal C$, giving
\begin{equation}
 S_\alpha(\ell)=\ell\log2-\frac{1}{2\pi i}\oint_{\mathcal C}
 e_\alpha'(\lambda)\log\!\left[\lambda^{-\ell}D_\ell(\lambda)\right]d\lambda.
 \label{eq:entropy-parts}
\end{equation}
The explicit term is $\ell e_\alpha(0)=\ell\log2$. It follows from subtracting $\ell/\lambda$ from the logarithmic derivative in Eq.~\eqref{eq:entropy-contour}. Appendix~\ref{app:entropy-extraction} gives the resulting residue and branch-cut decomposition.

For $\theta\in[-\pi,\pi]$, the even step symbol of the determinant is
\begin{equation}
 g(e^{i\theta};\lambda)=
 \begin{cases}
  \lambda-1,& |\theta|<k_F,\\
  \lambda+1,& |\theta|>k_F.
 \end{cases}
 \label{eq:step-symbol}
\end{equation}
Define its Fourier coefficients by
\begin{equation}
 f_m=\frac{1}{2\pi}\int_{-\pi}^{\pi}
 g(e^{i\theta};\lambda)e^{-im\theta}\,d\theta.
\label{eq:symbol-fourier-coefficients}
\end{equation}
Here $m\in\mathbb Z$.
If the matrix indices start at zero, then
\begin{equation}
 \left[(\lambda+1)I_\ell-2C\right]_{j+1,k+1}
 =f_{j-k}-f_{j+k+2}.
\label{eq:toeplitz-hankel}
\end{equation}
Equation~\eqref{eq:toeplitz-hankel} holds for $j,k=0,\ldots,\ell-1$.
The exact identity for an even symbol maps this Toeplitz-plus-Hankel determinant to a Hankel determinant on $[-1,1]$~\cite{DeiftItsKrasovsky2011,BasorEhrhardt2001,BasorEhrhardt2002}. The same $x=\cos\theta$ reduction was used for Toeplitz$\pm$Hankel minors in Ref.~\cite{GarciaGarciaTierz2020}. For a weight $w$, define
\begin{equation}
 H_\ell[w]=\det\left[\int_{-1}^{1}x^{j+k}w(x)\,dx\right]_{j,k=0}^{\ell-1}.
 \label{eq:Hankel-definition}
\end{equation}
Then
\begin{equation}
 D_\ell(\lambda)=\frac{2^{\ell^2}}{\pi^\ell}H_\ell[w_\lambda],
 \qquad
 w_\lambda(x)=\sqrt{1-x^2}\,g(e^{i\theta(x)};\lambda),
 \quad \theta(x)=\arccos x\in[0,\pi].
 \label{eq:even-symbol-identity}
\end{equation}
The prefactor is essential. Indeed,
\begin{equation}
 H_\ell\!\left[\sqrt{1-x^2}\right]=\pi^\ell2^{-\ell^2}.
 \label{eq:chebyshev-normalization}
\end{equation}

Put $x_0=\cos k_F$. For $\eta\in\mathbb C$, define
\begin{equation}
 \Xi_\eta(x;x_0)=
 \begin{cases}
  1,&-1<x<x_0,\\
  \eta,&x_0<x<1.
 \end{cases}
 \label{eq:jump-factor}
\end{equation}
The weight is
\begin{equation}
 w_\lambda(x)=(\lambda+1)\sqrt{1-x^2}\,
 \Xi_{(\lambda-1)/(\lambda+1)}(x;x_0).
\label{eq:physical-weight}
\end{equation}
The values of the step functions at $|\theta|=k_F$, $x=x_0$, and $x=\pm1$ do not affect the Fourier coefficients or moments. Equation~\eqref{eq:physical-weight} at $\lambda=-1$ is understood by continuous extension.
Introduce the jump parameter
\begin{equation}
 \beta(\lambda)=\frac{1}{2\pi i}
 \log\frac{\lambda+1}{\lambda-1},
 \qquad \beta(\lambda)\longrightarrow0
 \quad(\lambda\to\infty),
 \label{eq:beta-definition}
\end{equation}
where the logarithm is the branch analytic on $\mathbb C\setminus[-1,1]$ that vanishes at infinity. The jump ratio is $e^{-2\pi i\beta}$. Define
\begin{equation}
 w_\beta(x)=\sqrt{1-x^2}\,
 \Xi_{e^{-2\pi i\beta}}(x;x_0),
 \label{eq:beta-weight}
\end{equation}
and put
\begin{equation}
 H_\ell(\beta)=H_\ell[w_\beta],
 \qquad
 \cD_\ell(\beta)=\frac{H_\ell(\beta)}{H_\ell(0)}.
 \label{eq:normalized-Hankel}
\end{equation}
Equations~\eqref{eq:even-symbol-identity}--\eqref{eq:chebyshev-normalization} give the exact normalization
\begin{equation}
 D_\ell(\lambda)=(\lambda+1)^\ell\cD_\ell(\beta(\lambda)).
 \label{eq:exact-physical-normalization}
\end{equation}
The scalar factor belongs to the zero-frequency part. Therefore the first-harmonic part of $\log D_\ell$ is exactly that of $\log\cD_\ell$.

The normalization also gives an exact nonvanishing statement. If $C_n$ is the correlation matrix of the first $n$ sites and $\nu_j^{(n)}\in(0,1)$ are its eigenvalues, then
\[
 \cD_n(\beta)
 =\det\!\left[I_n+(e^{-2\pi i\beta}-1)C_n\right]
 =\prod_{j=1}^n\left[1-\nu_j^{(n)}+\nu_j^{(n)}e^{-2\pi i\beta}\right].
\]
A factor could vanish only if $e^{-2\pi i\beta}$ were negative real, which is impossible for $|\Re\beta|<1/2$. Hence $H_n(\beta)\ne0$ for every $n$ throughout this simply connected strip, and $\log\cD_n$ has a unique holomorphic branch normalized by $\log\cD_n(0)=0$.

\begin{figure}[H]
 \centering
 \includegraphics[width=0.94\linewidth]{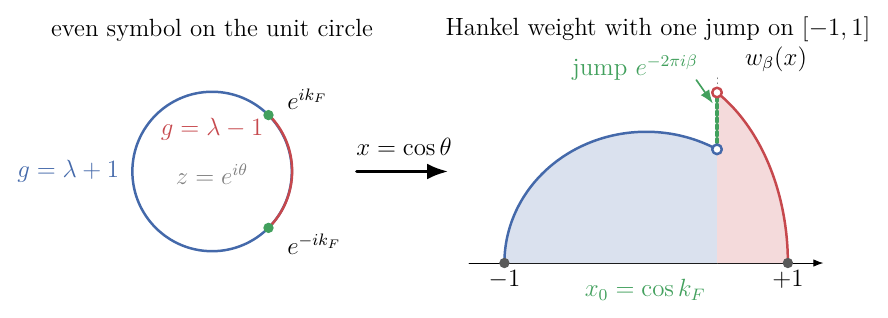}
 \caption{The even symbol has two jumps at $e^{\pm ik_F}$. The map $x=\cos\theta$ folds them into a single internal jump at $x_0=\cos k_F$. The weight has fixed square-root behavior at the endpoints $x=\pm1$. The plotted profile illustrates a positive real jump. On the entropy contour the jump is generally complex and is parametrized by $e^{-2\pi i\beta}$.}
 \label{fig:determinant-reduction}
\end{figure}

\section{Hankel determinant with a complex jump}
\label{sec:complex-jump}

For every integer $n\ge1$, let $H_n(\beta)$ and $\cD_n(\beta)$ denote the quantities in Eq.~\eqref{eq:normalized-Hankel} with $\ell$ replaced by $n$. In this section $n\to\infty$ through positive integers. The entropy contour samples complex values of $\lambda$, so the jump parameter $\beta(\lambda)$ and hence the Hankel jump are generally complex. We need asymptotics for $\cD_n(\beta)$ throughout the strip $|\Re\beta|<1/2$. The Riemann--Hilbert construction for a positive internal jump and the corresponding recurrence-coefficient calculation are established in Refs.~\cite{FoulquieMorenoMartinezFinkelshteinSousa2011,FoulquieMorenoMartinezFinkelshteinSousaJAT2010}, while complex Fisher--Hartwig parameters for Jacobi-type Hankel determinants are treated in Ref.~\cite{CharlierGharakhloo2021}. Here we continue the recurrence-coefficient calculation to compact subsets of the full strip and obtain the uniform subleading estimates needed below. The endpoint behavior is fixed, while the local model solution, or parametrix, at $x_0$ is expressed through confluent hypergeometric functions. In the large-$n$ expansion, the subscripts $+$ and $-$ denote the coefficients proportional to $e^{-i(2n+1)k_F}$ and $e^{i(2n+1)k_F}$, respectively.

We define the two coefficients algebraically by tagging the elementary factors $e^{-i(2n+1)k_F}$ and $e^{i(2n+1)k_F}$ with an auxiliary Laurent variable $\xi$ as $\xi$ and $\xi^{-1}$ before multiplying the asymptotic expansion, and then extracting the corresponding Laurent coefficients. For any tagged expansion $X$, we denote these coefficients by $X_+$ and $X_-$ and write $[X]^{(1)}=X_++X_-$. This algebraic extraction keeps them distinct even when different continuum wave numbers coincide modulo $2\pi$ on the integer lattice. If no such coincidence occurs, $[X]^{(1)}$ is the ordinary spatial $2k_F$ component. The two coefficients are evaluated on the same locally defined logarithm branch and are not alternative branches. Appendix~\ref{app:complex-RH} proves that the extraction commutes with the expansion and with finite differences. When applied to an exact observable, this notation denotes its corresponding coefficients in the uniform asymptotic expansion, rather than a Fourier projection at finite block length.

Throughout the asymptotic statements below, subscripts on $O(\cdot)$ list the fixed quantities on which the implicit constant may depend. In particular, $\delta$ fixes the uniform filling interval $k_F\in[\delta,\pi-\delta]$. It is not an asymptotic variable. The following theorem is the analytic input for the entropy calculation.

\begin{theorem}[First-harmonic components for a complex jump]
\label{thm:complex-jump}
Fix $0<\delta<\pi/2$ and let $K$ be a compact subset of the open strip
$\{\beta\in\mathbb C:|\Re\beta|<1/2\}$.
Uniformly for $k_F\in[\delta,\pi-\delta]$ and $\beta\in K$, put
\begin{equation}
 N_n=2n+1,
 \qquad \cL_n=4n\sin k_F.
 \label{eq:Nn-Ln}
\end{equation}
All powers of the positive quantities $n$ and $\cL_n$ below use the real logarithm. Define
\begin{align}
 \Psi_{n,+}(\beta)
 &=i e^{-iN_nk_F}\cL_n^{-1-2\beta}
   \frac{\Gamma(1+\beta)}{\Gamma(-\beta)},
 \label{eq:Psi-plus}\\
 \Psi_{n,-}(\beta)
 &=-i e^{iN_nk_F}\cL_n^{-1+2\beta}
   \frac{\Gamma(1-\beta)}{\Gamma(\beta)}.
 \label{eq:Psi-minus}
\end{align}
The Gamma ratios use the entire reciprocal $1/\Gamma$ and therefore extend holomorphically to $\beta=0$.
Use the holomorphic branch of $\log\cD_n$ on the strip normalized by $\log\cD_n(0)=0$, whose existence follows from the exact product above. The coefficients proportional to $e^{-iN_nk_F}$ and $e^{iN_nk_F}$ satisfy
\begin{align}
 \bigl(\log\cD_n(\beta)\bigr)_+
 &=\Psi_{n,+}(\beta)
   +e^{-iN_nk_F}n^{-2-2\beta}q_{n,+}(\beta;k_F),
 \label{eq:det-sector-plus}\\
 \bigl(\log\cD_n(\beta)\bigr)_-
 &=\Psi_{n,-}(\beta)
   +e^{iN_nk_F}n^{-2+2\beta}q_{n,-}(\beta;k_F),
 \label{eq:det-sector-minus}
\end{align}
where $q_{n,\pm}$ are locally holomorphic in $\beta$. Define
$\nabla q_{n,\pm}=q_{n+1,\pm}-q_{n,\pm}$. Then, for every fixed integer $m\ge0$,
\begin{equation}
 \nabla^m q_{n,\pm}
 =O_{K,\delta,m}(n^{-m})
 \label{eq:q-symbol-bounds}
\end{equation}
uniformly on $K\times[\delta,\pi-\delta]$.
\end{theorem}

We refer to the coincidence of different continuum wave numbers on the integer lattice as lattice aliasing.

We give the structure of the proof before the full verification in Appendix~\ref{app:complex-RH}. Starting from the local jump problem for a positive weight with the same square-root endpoint factors, we rewrite the construction without complex conjugation. Its jump matrices, hypergeometric solution, and matching coefficients are holomorphic in $\beta$ throughout $|\Re\beta|<1/2$. On a compact rectangle inside this strip, both off-diagonal matching terms tend uniformly to zero, so the final error problem is of small norm. This supplies the uniform recurrence estimates used here. Solvability itself already follows from exact nonvanishing.

Let $\pi_n(x;\beta)$ be the monic polynomials defined by the bilinear moment conditions
\begin{equation}
 \int_{-1}^{1}\pi_n(x;\beta)x^j w_\beta(x)\,dx=0,
 \qquad j=0,\ldots,n-1.
 \label{eq:complex-orthogonality}
\end{equation}
No complex conjugation enters these conditions. Their three-term recurrence is
\begin{equation}
 \pi_{n+1}(x;\beta)
 =(x-b_n(\beta))\pi_n(x;\beta)
 -a_n(\beta)^2\pi_{n-1}(x;\beta),
 \qquad \pi_{-1}=0,\quad \pi_0=1.
 \label{eq:orthogonal-polynomial-recurrence}
\end{equation}
Here $b_n(\beta)$ is the diagonal recurrence coefficient. For all sufficiently large $n$, $a_n(\beta)$ denotes the locally holomorphic square root fixed by $a_n(0)>0$ and $a_n(\beta)\to1/2$ as $n\to\infty$ on compact subsets of the strip. Only this coefficient enters the determinant identity used here. The residue at $x_0$ of the first matching coefficient between the local and outer parametrices determines the leading $+$ and $-$ coefficients of $a_n(\beta)-\frac12$:
\begin{equation}
 \left(a_n(\beta)-\frac12\right)_\pm
 =-\sin^2k_F\,\Psi_{n,\pm}(\beta)
 +\text{terms smaller by one additional power of }n.
 \label{eq:recurrence-summary}
\end{equation}
If $h_n=H_{n+1}/H_n$ denotes the bilinear normalization constant of the monic polynomial, then $a_n^2=h_n/h_{n-1}$. Together with $H_n(0)=\pi^n2^{-n^2}$, this gives the exact identity
\begin{equation}
 \Delta^2\log\cD_n(\beta)=\log\bigl(4a_n(\beta)^2\bigr),
 \qquad
 \Delta^2f_n=f_{n+1}-2f_n+f_{n-1}
 \quad\text{for any sequence }f_n.
 \label{eq:second-difference-summary}
\end{equation}
For fixed $s\in\mathbb C$, the action of $\Delta^2$ on either oscillatory factor is
\begin{equation}
 \Delta^2\!\left(e^{\mp i(2n+1)k_F}n^{-s}\right)
 =e^{\mp i(2n+1)k_F}n^{-s}
 \left[-4\sin^2k_F+O_s(n^{-1})\right].
 \label{eq:sector-multiplier}
\end{equation}
The leading multiplier $-4\sin^2k_F$ is bounded away from zero under the theorem's filling hypothesis. Inversion separately for the two exponential factors therefore determines Eqs.~\eqref{eq:det-sector-plus} and \eqref{eq:det-sector-minus} without changing the power of $n$. Appendix~\ref{app:complex-RH} organizes the matching terms and the terms in the Neumann series for the error problem by the auxiliary Laurent variable and inverse powers of $n$. It shows that the next contribution with the same exponential factor carries one additional factor $n^{-1}$.

\section{Boundary R\'enyi entropy oscillations}
\label{sec:renyi-oscillations}

We now insert the two determinant coefficients of Theorem~\ref{thm:complex-jump} into the exact entropy contour. For fixed noninteger $\alpha>1$, contour deformation isolates the finite set of residues of $e_\alpha'$ and leaves an integral along the exterior real cuts. Its extracted first harmonic is bounded using Appendix~\ref{app:positive-jump-uniform}. At integer index, exact factorization replaces the contour argument and sharpens the first-harmonic error.

Set $n=\ell$ in Theorem~\ref{thm:complex-jump} and put
\begin{equation}
 N=2\ell+1,
 \qquad \cL=4\ell\sin k_F,
 \qquad
 a_{\alpha,q}=\frac{2q-1}{\alpha},
 \quad q=1,2,\ldots,
 \label{eq:NL-definitions}
\end{equation}
and
\begin{equation}
 Q_{\alpha,q}=
 \frac{\Gamma\!\left(\frac12+\frac{2q-1}{2\alpha}\right)}
      {\Gamma\!\left(\frac12-\frac{2q-1}{2\alpha}\right)}.
 \label{eq:Q-definition}
\end{equation}
Here $a_{\alpha,q}>0$ is a decay power, so its contribution contains $\cL^{-a_{\alpha,q}}$. The index $q$ enumerates residues within the first harmonic, not different spatial harmonics, and $Q_{\alpha,q}$ is their coefficient. We set $\delta S_\alpha^{(1)}=[S_\alpha]^{(1)}$, with the extraction defined in Section~\ref{sec:complex-jump}. Unless stated otherwise, all estimates below are for $\ell\to\infty$ through positive integers.

\subsection{Entropy oscillations for fixed \texorpdfstring{$\alpha>1$}{alpha > 1}}

Let
\begin{equation}
 F_\alpha(\lambda)=
 \left(\frac{1+\lambda}{2}\right)^\alpha+
 \left(\frac{1-\lambda}{2}\right)^\alpha.
 \label{eq:Falpha}
\end{equation}
On the initial sheet fixed below Eq.~\eqref{eq:renyi-kernel}, the finite zeros of $F_\alpha$ are
\begin{equation}
 \lambda_{q,\pm}=\pm i\tan\frac{(2q-1)\pi}{2\alpha},
 \qquad 2q-1<\alpha.
 \label{eq:kernel-zeros}
\end{equation}
They are simple, and $e_\alpha'=(1-\alpha)^{-1}F_\alpha'/F_\alpha$ has residue $(1-\alpha)^{-1}$ at each one. Their jump parameters are
\begin{equation}
 \beta(\lambda_{q,+})=-\frac12+\frac{2q-1}{2\alpha},
 \qquad
 \beta(\lambda_{q,-})=\frac12-\frac{2q-1}{2\alpha}.
 \label{eq:beta-at-zeros}
\end{equation}
Both lie strictly inside the strip in Theorem~\ref{thm:complex-jump}.

Equations~\eqref{eq:NL-definitions} and \eqref{eq:beta-at-zeros} give the exact dictionary
\[
 a_{\alpha,q}
 =1+2\beta(\lambda_{q,+})
 =1-2\beta(\lambda_{q,-}).
\]
Thus $a_{\alpha,q}$ is the positive decay power obtained by evaluating the coefficient proportional to $e^{-iNk_F}$ at $\lambda_{q,+}$, or the coefficient proportional to $e^{iNk_F}$ at $\lambda_{q,-}$. It is not the Fisher--Hartwig jump parameter $\beta$.

\begin{theorem}[First-harmonic residues with decay power below one]
\label{thm:real-alpha}
Fix $\alpha>1$ and $0<\delta<\pi/2$. Uniformly for
$k_F\in[\delta,\pi-\delta]$, the sum of the coefficients proportional to $e^{-iNk_F}$ and $e^{iNk_F}$ in the R\'enyi entropy is
\begin{equation}
 \delta S_\alpha^{(1)}(\ell)=
 \frac{2\sin(Nk_F)}{1-\alpha}
 \sum_{\substack{q\ge1\\2q-1<\alpha}}
 Q_{\alpha,q}\cL^{-a_{\alpha,q}}
 +O_{\alpha,\delta}(\cL^{-1}) .
 \label{eq:main-real-alpha}
\end{equation}
Every displayed term is asymptotically separated from the remainder because $a_{\alpha,q}<1$.
\end{theorem}

Thus the calculation retains the finite sequence
\begin{equation}
 \cL^{-1/\alpha},\quad \cL^{-3/\alpha},\quad
 \cL^{-5/\alpha},\ldots,
 \qquad 2q-1<\alpha.
 \label{eq:power-sequence}
\end{equation}
The restriction $a_{\alpha,q}<1$ selects the finite residues on the initial sheet and the first-harmonic terms that decay more slowly than $\cL^{-1}$. It is a controlled truncation, not an assertion that formal generalized Fisher--Hartwig terms with $a_{\alpha,q}\ge1$ are absent. The statement applies at fixed $\alpha$ and is not uniform as $\alpha$ approaches an odd integer, where the next candidate decay power reaches one.

The exact decomposition used in the proof is Eq.~\eqref{eq:residue-decomposition}. The crossed poles are the zeros of $F_\alpha$ in Eq.~\eqref{eq:kernel-zeros}. The determinant zeros remain inside the spectral contour. The additional branch-cut integral is essential for noninteger $\alpha$. Its jump parameters run along the imaginary axis without remaining in a fixed compact set. Appendix~\ref{app:positive-jump-uniform} proves the required uniform estimate, and Appendix~\ref{app:entropy-extraction} combines it with an exact bound on the tail before extracting the Laurent coefficients.

At $\lambda_{q,+}$, Eq.~\eqref{eq:Psi-plus} gives
\begin{equation}
 i e^{-iNk_F}Q_{\alpha,q}\cL^{-a_{\alpha,q}},
 \label{eq:positive-zero-residue}
\end{equation}
while Eq.~\eqref{eq:Psi-minus} at $\lambda_{q,-}$ gives the conjugate term,
\begin{equation}
 -i e^{iNk_F}Q_{\alpha,q}\cL^{-a_{\alpha,q}}.
 \label{eq:negative-zero-residue}
\end{equation}
Their sum is $2Q_{\alpha,q}\cL^{-a_{\alpha,q}}\sin(Nk_F)$. At either pole, the term with the other exponential factor is $O(\cL^{-2+a_{\alpha,q}})$, and the next term with the displayed exponential factor is $O(\cL^{-1-a_{\alpha,q}})$.

\subsection{Integer index: exact determinant factorization}

For an integer $p\ge2$, define
\begin{equation}
 P_p(x)=\left(\frac{1+x}{2}\right)^p+
        \left(\frac{1-x}{2}\right)^p,
 \label{eq:Pp}
\end{equation}
and
\begin{equation}
 c_p=\begin{cases}
 2^{1-p},&p\text{ even},\\
 p\,2^{1-p},&p\text{ odd}.
 \end{cases}
 \label{eq:cp}
\end{equation}
Its finite roots form conjugate pairs
\begin{equation}
 r_{q,\pm}=\pm i\tan\frac{(2q-1)\pi}{2p},
 \qquad q=1,\ldots,\left\lfloor\frac p2\right\rfloor.
 \label{eq:integer-roots}
\end{equation}
Consequently,
\begin{equation}
 P_p(x)=c_p\prod_{q=1}^{\lfloor p/2\rfloor}
 (x-r_{q,+})(x-r_{q,-}).
 \label{eq:Pp-factorization-main}
\end{equation}
Taking the product of Eq.~\eqref{eq:Pp-factorization-main} over the transformed correlation eigenvalues $x_j$ gives the exact determinant identity
\begin{equation}
 S_p(\ell)=\frac{1}{1-p}\left[
 \ell\log c_p+2\sum_{q=1}^{\lfloor p/2\rfloor}
 \log\left|D_\ell\!\left(
 i\tan\frac{(2q-1)\pi}{2p}\right)\right|
 \right].
 \label{eq:integer-factorization-main}
\end{equation}
In the notation of Eqs.~\eqref{eq:NL-definitions} and \eqref{eq:Q-definition},
$a_{p,q}=(2q-1)/p$ and $Q_{p,q}$ is $Q_{\alpha,q}$ evaluated at $\alpha=p$.

\begin{corollary}[Integer R\'enyi index]
\label{cor:integer-p}
Fix an integer $p\ge2$ and $0<\delta<\pi/2$. Uniformly for
$k_F\in[\delta,\pi-\delta]$, the first-harmonic projection is
\begin{equation}
 \delta S_p^{(1)}(\ell)=
 \frac{2\sin(Nk_F)}{1-p}
 \sum_{q=1}^{\lfloor p/2\rfloor}
 Q_{p,q}\cL^{-(2q-1)/p}
 +O_{p,\delta}(\cL^{-1-1/p}) .
 \label{eq:main-integer-p}
\end{equation}
\end{corollary}

There is no residual entropy contour in Eq.~\eqref{eq:integer-factorization-main}. For the $q$th conjugate pair, the retained term is $\cL^{-a_{p,q}}$, the next term with the same exponential factor is $\cL^{-1-a_{p,q}}$, and the term with the other exponential factor is $\cL^{-(2-a_{p,q})}$. The slowest omitted power among the finitely many roots is $\cL^{-1-1/p}$. This proves the stated error after summation. The detailed algebra is in Appendix~\ref{app:integer-factorization}.

For example,
\begin{equation}
 \delta S_4^{(1)}(\ell)
 =-\frac23\sin(Nk_F)
 \left[Q_{4,1}\cL^{-1/4}+Q_{4,2}\cL^{-3/4}\right]
 +O_\delta(\cL^{-5/4}).
 \label{eq:p4-example}
\end{equation}
The error estimates in Eqs.~\eqref{eq:main-real-alpha} and \eqref{eq:main-integer-p} apply only to the extracted $2k_F$ component. They do not bound the remainder of the full entropy, whose smooth and higher-harmonic corrections can be larger.

\subsection{Physical interpretation and relation to Fisher--Hartwig predictions}

The leading $q=1$ term of Eq.~\eqref{eq:main-real-alpha} is
\begin{equation}
 \delta S_{\alpha,q=1}^{(1)}(\ell)=
 \frac{2}{1-\alpha}
 \frac{\Gamma(\frac12+\frac{1}{2\alpha})}
      {\Gamma(\frac12-\frac{1}{2\alpha})}
 \cL^{-1/\alpha}\sin[(2\ell+1)k_F].
 \label{eq:leading-physical-oscillation}
\end{equation}
This is the leading open-chain parity or Friedel correction obtained by Fagotti and Calabrese~\cite{FagottiCalabrese2010JSM}. The leading boundary decay power is $1/\alpha$, rather than the periodic-chain value $2/\alpha$~\cite{CalabreseCampostriniEsslerNienhuis2010PRL,CalabreseEssler2010JSM}, because boundary reflection folds the two Fermi points into the single jump in Fig.~\ref{fig:determinant-reduction}. In field theory, the corresponding bulk two-point structure is replaced by a boundary one-point contribution~\cite{CalabreseCardy2010JSM,CalabreseCampostriniEsslerNienhuis2010PRL}.

The $q\ge2$ members of this first-harmonic sequence are already present in the generalized Fisher--Hartwig expansion of Fagotti and Calabrese~\cite{FagottiCalabrese2010JSM}. Here the two determinant coefficients are derived directly by Riemann--Hilbert analysis and inserted into the exact entropy contour. This includes every finite zero on the initial sheet that produces a $2k_F$ term with decay power below one. At integer index the same method gives the $O(\cL^{-1-1/p})$ error within that component.

\section{Charged moments at fixed flux}
\label{sec:charged}

Laflorencie and Rachel studied spin-resolved entanglement spectra and the entropy contributions of individual spin sectors in critical XXZ chains~\cite{LaflorencieRachel2014}.

The charged moment keeps the counting flux rather than setting it to zero:
\begin{equation}
 Z_\alpha(\phi)=\Tr\!\left(\rho_A^{\,\alpha}\,e^{i\phi Q_A}\right),
 \qquad Q_A=\sum_{j=1}^{\ell}c_j^\dagger c_j,
 \qquad \phi\in(-\pi,\pi).
 \label{eq:charged-moment-definition}
\end{equation}
Here $Q_A$ is the fermion number in the block and $\phi$ is its conjugate counting flux. The uncharged quantity is $Z_\alpha(0)=\Tr\rho_A^\alpha$. We use the branch of $\log Z_\alpha(\phi)$ continued from its real value at $\phi=0$. Throughout this section, $N=2\ell+1$, $\cL=4\ell\sin k_F$, and $[\,\cdot\,]^{(1)}$ is the first-harmonic extraction defined in Section~\ref{sec:complex-jump}.
We also use the coefficient function
\[
 Q(a)=\frac{\Gamma\!\left(\frac{1+a}{2}\right)}
           {\Gamma\!\left(\frac{1-a}{2}\right)}.
\]

The global $U(1)$ symmetry implies $[\rho_A,Q_A]=0$, so Fourier coefficients in the flux isolate fixed-charge sectors. If $\Pi_m$ projects onto charge $m\in\{0,\ldots,\ell\}$, define
\begin{equation}
 \mathcal Z_\alpha(m)=\Tr(\rho_A^\alpha\Pi_m)
 =\int_{-\pi}^{\pi}\frac{d\phi}{2\pi}
 e^{-im\phi}Z_\alpha(\phi),
 \qquad
 S_\alpha(m)=\frac{1}{1-\alpha}
 \log\frac{\mathcal Z_\alpha(m)}{\mathcal Z_1(m)^\alpha}.
 \label{eq:sre-definition}
\end{equation}
The first formula defines $\mathcal Z_\alpha(m)$ for $\alpha>0$. The expression for $S_\alpha(m)$ holds for $\alpha\ne1$, with the continuous limit giving the symmetry-resolved von Neumann entropy. Moreover, $\mathcal Z_1(m)$ is the probability distribution of the block charge~\cite{GoldsteinSela2018,BonsignoriRuggieroCalabrese2019,MurcianoDiGiulioCalabrese2020,BonsignoriCalabrese2021,CastroAlvaredoSantamariaSanz2025}. For the boundary conformal field theory of symmetry resolution see Ref.~\cite{KusukiMurcianoOoguriPal2023}. This section evaluates the $2k_F$ component of $\log Z_\alpha(\phi)$ at fixed $\phi$. It combines Theorem~\ref{thm:complex-jump} with the exact counting identities, the product over shifted fluxes, and the contour argument for the charged factor of one correlation eigenvalue. It does not perform the Fourier integral in Eq.~\eqref{eq:sre-definition}.

\subsection{Full counting statistics at \texorpdfstring{$\alpha=1$}{alpha = 1}}

At $\alpha=1$, $Z_1(\phi)$ is the characteristic function of the block charge, also called the generating function of its full counting statistics. For the free-fermion ground state it is a single correlation determinant. From $\nu_j=(1+x_j)/2$,
\begin{equation}
 Z_1(\phi)=\prod_{j=1}^{\ell}\left[1-\nu_j+\nu_je^{i\phi}\right]
 =\left(\frac{1-e^{i\phi}}{2}\right)^{\!\ell}
 D_\ell\!\left(\lambda_*(\phi)\right),
 \qquad
 \lambda_*(\phi)=i\cot\frac{\phi}{2}.
 \label{eq:fcs-identity}
\end{equation}
The identity is exact for $\phi\ne0$, and $\phi=0$ is understood as the limit $\lambda_*\to\infty$, where $Z_1(0)=1$. Moreover
\begin{equation}
 \frac{\lambda_*+1}{\lambda_*-1}=e^{-i\phi},
 \qquad
 \beta\bigl(\lambda_*(\phi)\bigr)=-\frac{\phi}{2\pi}.
 \label{eq:fcs-beta}
\end{equation}
The counting factor $e^{i\phi}$ is precisely the Hankel jump ratio of Section~\ref{sec:determinant-reduction}. For $\phi$ in a compact subset of $(-\pi,\pi)$, the jump parameter stays in a compact subset of the strip $|\Re\beta|<1/2$, and $|\Re\beta|\to1/2$ exactly as $\phi\to\pm\pi$.

The subscripts on $(\log Z_1)_\pm$ denote the coefficients proportional to $e^{-iNk_F}$ and $e^{iNk_F}$, respectively.
\begin{corollary}[First harmonic in full counting statistics]
\label{cor:fcs}
Fix $0<\delta<\pi/2$ and let $K_\phi$ be a compact subset of $(-\pi,\pi)$. Uniformly for $\phi\in K_\phi$ and $k_F\in[\delta,\pi-\delta]$,
\begin{align}
 \bigl(\log Z_1(\phi)\bigr)_+
 &= i e^{-iNk_F}\,Q(1-\phi/\pi)\,\cL^{-(1-\phi/\pi)}
   +O_{K_\phi,\delta}\bigl(\cL^{-(2-\phi/\pi)}\bigr),
 \label{eq:fcs-plus}\\
 \bigl(\log Z_1(\phi)\bigr)_-
 &= -i e^{iNk_F}\,Q(1+\phi/\pi)\,\cL^{-(1+\phi/\pi)}
   +O_{K_\phi,\delta}\bigl(\cL^{-(2+\phi/\pi)}\bigr).
 \label{eq:fcs-minus}
\end{align}
\end{corollary}

The logarithm of the scalar prefactor in Eq.~\eqref{eq:fcs-identity} is affine in $\ell$ and belongs to the zero-frequency part. The leading powers and gamma coefficients in Eqs.~\eqref{eq:fcs-plus}--\eqref{eq:fcs-minus} are the $\alpha=1$ specialization of the two leading boundary terms in Ref.~\cite{BonsignoriCalabrese2021}. Corollary~\ref{cor:fcs} obtains them from Theorem~\ref{thm:complex-jump} at $\beta=-\phi/(2\pi)$ and supplies errors in the two coefficients, uniform on compact flux subsets. For $\phi>0$, the term proportional to $e^{-iNk_F}$ decays more slowly than $\cL^{-1}$, whereas the term proportional to $e^{iNk_F}$ decays more rapidly. The flux dependence therefore cannot be absorbed into an $\ell$-independent amplitude. Expanding at small $\phi$ recovers the Friedel tail of the mean block charge: $\log Z_1(\phi)=i\phi\langle Q_A\rangle+O(\phi^2)$, and the two coefficients reproduce the oscillation $\cos(Nk_F)/(\pi\cL)$ of $\langle Q_A\rangle$ with the exact coefficient.

\subsection{Charged moments for fixed \texorpdfstring{$\alpha>1$}{alpha > 1}}

For general real $\alpha>1$, $Z_\alpha(\phi)$ is not obtained from a single evaluation of $D_\ell(\lambda)$ with $\beta=-\phi/(2\pi)$. Instead, define the charged factor for one correlation eigenvalue,
\begin{equation}
 F_{\alpha,\phi}(\lambda)=
 \left(\frac{1+\lambda}{2}\right)^{\!\alpha}e^{i\phi}+
 \left(\frac{1-\lambda}{2}\right)^{\!\alpha},
 \qquad
 \log Z_\alpha(\phi)=\sum_{j=1}^{\ell}\log F_{\alpha,\phi}(x_j).
 \label{eq:charged-kernel}
\end{equation}
The powers use the branch convention fixed below Eq.~\eqref{eq:renyi-kernel}, and the factor logarithms are continued from their values at $\phi=0$ as specified after Eq.~\eqref{eq:charged-moment-definition}. Define
\begin{equation}
 \begin{gathered}
 a_{q,-}\equiv a_{q,-}(\alpha,\phi)
 =\frac{2q-1-\phi/\pi}{\alpha},
 \qquad
 a_{q,+}\equiv a_{q,+}(\alpha,\phi)
 =\frac{2q-1+\phi/\pi}{\alpha},
 \qquad q=1,2,\ldots,\\
 \varSigma_\mp=\sum_{\substack{q\ge1\\ a_{q,\mp}<1}}
 Q(a_{q,\mp})\,\cL^{-a_{q,\mp}}.
 \end{gathered}
 \label{eq:aq-Qa}
\end{equation}
The signs in $a_{q,\mp}$ and $\varSigma_\mp$ label the two flux shifts. At $\phi=0$, both powers reduce to $a_{\alpha,q}$ and $Q_{\alpha,q}=Q(a_{\alpha,q})$. The inequality $a_{q,\mp}<1$ selects the zeros that are finite on the initial sheet. On that sheet, $F_{\alpha,\phi}$ has the two families
\begin{equation}
 \mu_{q,-}=i\tan\frac{\pi a_{q,-}}{2},
 \qquad
 \mu_{q,+}=-i\tan\frac{\pi a_{q,+}}{2},
 \qquad a_{q,\mp}<1.
 \label{eq:charged-zeros}
\end{equation}
Their jump parameters are
\begin{equation}
 \beta(\mu_{q,-})=\frac{a_{q,-}-1}{2},
 \qquad
 \beta(\mu_{q,+})=\frac{1-a_{q,+}}{2}.
 \label{eq:charged-beta}
\end{equation}
The subscript records the sign of the flux shift. The zero $\mu_{q,-}$ lies in the upper half-plane and contributes to the coefficient of $e^{-iNk_F}$, while $\mu_{q,+}$ lies in the lower half-plane and contributes to the coefficient of $e^{iNk_F}$. At $\phi=0$, $\mu_{q,-}=\lambda_{q,+}$ and $\mu_{q,+}=\lambda_{q,-}$ in Eq.~\eqref{eq:kernel-zeros}.

\begin{theorem}[First-harmonic expansion of $\log Z_\alpha(\phi)$ at fixed flux]
\label{thm:charged}
Fix $\alpha>1$, $0<\delta<\pi/2$, and $\phi\in(-\pi,\pi)$ with $a_{q,\mp}\ne1$ for every $q\ge1$. Uniformly for $k_F\in[\delta,\pi-\delta]$, the sum of the two first-harmonic coefficients is
\begin{equation}
 [\log Z_\alpha(\phi)]^{(1)}
 = i e^{-iNk_F}\,\varSigma_-
 - i e^{iNk_F}\,\varSigma_+
 + O_{\alpha,\phi,\delta}(\cL^{-1}).
 \label{eq:main-charged}
\end{equation}
Every displayed term is asymptotically separated from the remainder because its decay power is below one.
\end{theorem}

For noninteger $\alpha$, the proof uses the contour argument of Appendix~\ref{app:entropy-extraction} with the factor $F_{\alpha,\phi}$. Each zero has residue one. Normalization contributes $\ell\log F_{\alpha,\phi}(0)$ to the zero-frequency part, and the two branch cuts acquire different weights. Appendix~\ref{app:charged} gives these weights in Eq.~\eqref{eq:charged-cut-weight} and their bound. Equation~\eqref{eq:charged-decomposition} there is the charged analogue of Eq.~\eqref{eq:residue-decomposition}. At integer index, the result follows from the factorization in Corollary~\ref{cor:charged-integer}. At $\phi=0$ the two families merge into conjugate pairs, and Eq.~\eqref{eq:main-charged} reduces to $(1-\alpha)$ times Eq.~\eqref{eq:main-real-alpha}, consistent with $\log Z_\alpha(0)=(1-\alpha)S_\alpha(\ell)$. For odd integer $\alpha$, one decay power equals one at $\phi=0$. The uncharged case is therefore covered directly by Theorem~\ref{thm:real-alpha}, whose strict inequality omits the root whose decay power equals one. When some $a_{q,\mp}$ crosses one, the corresponding zero reaches infinity, exactly as in the uncharged case when $\alpha$ crosses an odd integer, and its coefficient vanishes continuously: $Q(a)\to0$ as $a\to1^-$. 

Since every $Q(a)$ in the sums is real, the displayed residue terms have the following real and imaginary parts,
\begin{equation}
 \begin{aligned}
 \Re\,[\log Z_\alpha(\phi)]^{(1)}
 &=\sin(Nk_F)\,(\varSigma_-+\varSigma_+)
 +O_{\alpha,\phi,\delta}(\cL^{-1}),\\
 \Im\,[\log Z_\alpha(\phi)]^{(1)}
 &=\cos(Nk_F)\,(\varSigma_--\varSigma_+)
 +O_{\alpha,\phi,\delta}(\cL^{-1}).
 \end{aligned}
 \label{eq:quadratures}
\end{equation}
At $\phi=0$, the two terms have equal decay powers and coefficients, so only their sum is visible in the entropy. At nonzero flux their decay powers differ. This is tested directly in Section~\ref{sec:numerics}.

\subsection{Integer index: exact shifted-flux product}

For an integer R\'enyi index $\alpha=p\ge2$, the charged factor of one correlation eigenvalue gives an exact product of $p$ factors $Z_1$ at shifted fluxes. Put
\begin{equation}
 \phi_r=\frac{\phi+(2r+1-p)\pi}{p},
 \qquad r=0,\ldots,p-1.
 \label{eq:twist-angles}
\end{equation}
All $\phi_r$ lie in $(-\pi,\pi)$ for $|\phi|<\pi$. Pairing the $p$th roots of $-e^{i\phi}$ gives the exact product identity over shifted fluxes,
\begin{equation}
 Z_p(\phi)=\prod_{r=0}^{p-1}Z_1(\phi_r),
 \label{eq:twist-factorization}
\end{equation}
proved in Appendix~\ref{app:charged}. The broader free-fermion relation between entanglement and full counting statistics was established in Refs.~\cite{KlichLevitov2009,SongBipartite2012}. The exact product at the discrete angles $\phi_r$, including nonzero flux, appears in Eq.~(2.37) of Ref.~\cite{CapizziMurcianoCalabrese2023}. The same shifted replica fluxes also appear in the construction using composite twist fields~\cite{GoldsteinSela2018,HorvathCalabreseCastroAlvaredo2022} in Ref.~\cite{LiDupaysRuggiero2026}. Here we combine the known factorization with the two first-harmonic coefficients of the boundary determinant. Because every factor satisfies Corollary~\ref{cor:fcs} and no contour remains, every term that dominates the common remainder can be retained. We define $\log Z_p(\phi)$ as the sum of the factor logarithms, with the branches fixed in Appendix~\ref{app:charged} before the coefficients are extracted.

\pagebreak
\begin{corollary}[Integer index: expansion of $\log Z_p(\phi)$ at fixed flux]
\label{cor:charged-integer}
Fix an integer $p\ge2$, $\phi\in(-\pi,\pi)$, and $0<\delta<\pi/2$. Set $\alpha=p$ and define the decay powers for integer $p$
$a_{q,\mp}^{[p]}=(2q-1\mp\phi/\pi)/p$ for $q=1,2,\ldots$. Put $\kappa_p(\phi)=(1-|\phi|/\pi)/p$ and define the extended sums
\begin{equation}
 \varSigma_\mp^{[p]}=
 \sum_{\substack{q\ge1\\a_{q,\mp}^{[p]}<1+\kappa_p(\phi)}}
 Q(a_{q,\mp}^{[p]})\,\cL^{-a_{q,\mp}^{[p]}}.
 \label{eq:extended-integer-sums}
\end{equation}
Uniformly for $k_F\in[\delta,\pi-\delta]$,
\begin{equation}
 [\log Z_p(\phi)]^{(1)}
 =i e^{-iNk_F}\varSigma_-^{[p]}
 -i e^{iNk_F}\varSigma_+^{[p]}
 +O_{p,\phi,\delta}(\cL^{-1-\kappa_p(\phi)}).
 \label{eq:main-charged-integer}
\end{equation}
\end{corollary}

At $\phi=0$, $\kappa_p=1/p$ and the corollary reduces to $(1-p)$ times Eq.~\eqref{eq:main-integer-p}. The decay power $a=1$ appears for odd $p$ only, with coefficient $Q(1)=0$. For odd $p$ and $0<|\phi|<\pi/2$, the extended sum retains a term with decay power $1+|\phi|/(p\pi)$. It decays faster than $\cL^{-1}$ and is retained by exact factorization rather than by the estimate for the remaining contour. For $p=4$ and $\phi=\pi/2$,
\begin{equation}
\begin{aligned}
 [\log Z_4(\pi/2)]^{(1)}
 ={}& i e^{-iNk_F}\left[Q(\tfrac18)\cL^{-1/8}+Q(\tfrac58)\cL^{-5/8}\right]\\
 &- i e^{iNk_F}\left[Q(\tfrac38)\cL^{-3/8}+Q(\tfrac78)\cL^{-7/8}\right]
 +O_{\delta}(\cL^{-9/8}).
\end{aligned}
 \label{eq:p4-charged-example}
\end{equation}
Bonsignori and Calabrese obtained the two leading boundary terms, with decay powers $(1\mp\phi/\pi)/p$, and stated that all enclosed residues should in principle be summed, but retained only the singularities nearest the real axis~\cite{BonsignoriCalabrese2021}. For a bulk interval in an infinite chain, Fraenkel and Goldstein used the generalized Fisher--Hartwig expansion to evaluate contributions from several flux-dependent residues through $o(\cL^{-1})$~\cite{FraenkelGoldstein2020JSM}. Their first-harmonic powers are $2(2q-1\mp\phi/\pi)/p$, displaying the expected factor of two relative to the boundary powers. In the boundary geometry, Theorem~\ref{thm:charged} evaluates every first-harmonic residue on the initial sheet with $a_{q,\mp}<1$ and controls the remaining contour. At integer index, Corollary~\ref{cor:charged-integer} additionally retains every factorization term that decays more slowly than its common error, including the term with decay power greater than one described above.

See Refs.~\cite{AbanovIvanovQian2011,IvanovAbanovCheianov2013} for branch sums in segment full counting statistics. In the periodic-chain analogue, Eqs.~(3.30)--(3.31) of Ref.~\cite{EstienneIkhlefMorinDuchesne2021} decompose the lattice composite twist operator into shifted-charge composite twist fields and descendants with oscillatory lattice factors. Jones proved leading charged-moment asymptotics with explicit error estimates for bulk intervals in a broader class of critical free-fermion chains~\cite{Jones2022}. His comparison with the complete subleading Fisher--Hartwig expansion uses additional conjectural Toeplitz information and an explicit uniformity assumption for the contour integral.

\subsection{Flux dependence and range of validity}

A fixed-flux expansion is an input to symmetry resolution, but it is not yet a fixed-charge asymptotic. In the zero-frequency, smooth part, symmetry resolution is controlled by a Gaussian saddle-point analysis and by equipartition, namely the leading equality of the R\'enyi entropies in the typical charge sectors~\cite{XavierAlcarazSierra2018,BonsignoriRuggieroCalabrese2019,BonsignoriCalabrese2021}. Applying that Gaussian approximation to the zero-flux oscillation would multiply its two exponential terms by the same factor. It cannot reproduce Eq.~\eqref{eq:main-charged}, because the two terms acquire the opposite linear shifts $\mp\phi/(\pi\alpha)$ in their decay powers, as the leading pair of Ref.~\cite{BonsignoriCalabrese2021} already exhibits. For fixed $\phi\ne0$ with all $a_{q,\mp}\ne1$, the slower of the two leading terms has decay power $(1-|\phi|/\pi)/\alpha$, rather than $1/\alpha$ at zero flux. Applying the formula for the two leading terms with the same estimator over the same finite window gives the $p=2$ enhancement tested with no fitted parameter in Section~\ref{sec:numerics}.

The expansion for fixed real $\alpha>1$ is pointwise at fixed $\phi\in(-\pi,\pi)$ with all $a_{q,\mp}\ne1$. The statement at integer R\'enyi index is covered separately by Corollary~\ref{cor:charged-integer}. The Fourier transform defining $\mathcal Z_\alpha(m)$ uses the complete charged moment $Z_\alpha(\phi)$ over $[-\pi,\pi]$, whereas Theorem~\ref{thm:charged} controls only the two extracted first-harmonic coefficients of $\log Z_\alpha(\phi)$ at each fixed admissible $\phi$. These pointwise estimates do not provide the uniform control near the exceptional fluxes and the endpoints needed to justify termwise Fourier projection. As $|\Re\beta|\to1/2$ at the endpoints, the Fisher--Hartwig contributions reorganize. Analysis uniform at the endpoints and fixed-charge projection are therefore outside this paper.

\section{Numerical tests of subleading terms and flux enhancement}
\label{sec:numerics}

In this section, $N=2\ell+1$ and $\cL=4\ell\sin k_F$, as in Eq.~\eqref{eq:NL-definitions}.
Equation~\eqref{eq:p4-example} contains a subleading $\cL^{-3/4}$ term at the same spatial frequency as the leading $\cL^{-1/4}$ term, and Eq.~\eqref{eq:p4-charged-example} contains the subleading $\cL^{-5/8}$ term in addition to the two leading charged terms. To resolve them, we use $k_F=1.4$, where the $2k_F$, $4k_F$, and $6k_F$ components remain distinct on the lattice. Half filling is unsuitable for this test because the $6k_F$ harmonic aliases with the parity oscillation at precisely the order being measured.

The fits use all 601 integer block lengths $\ell=180,181,\ldots,780$. For each $\ell$, the correlation matrix in Eq.~\eqref{eq:correlation-matrix} is diagonalized and $S_4(\ell)$ is evaluated from Eq.~\eqref{eq:entropy-spectrum}. With $B_q=-(2/3)Q_{4,q}$, we subtract the universal logarithm and the predicted leading first-harmonic term,
\begin{equation}
 R_\ell=S_4(\ell)-\frac{5}{48}\log\ell
 -B_1\cL^{-1/4}\sin(Nk_F).
 \label{eq:p4-residual}
\end{equation}
A local weighted fit uses centers $\ell_0=240,260,\ldots,720$, the window $|\ell-\ell_0|\le60$, and weights $\exp[-(\ell-\ell_0)^2/(2\times20^2)]$. It separates a smooth cubic background and the sine and cosine quadratures of the first three spatial harmonics, taken in the variable $Nk_F$ with $N=2\ell+1$: the oscillatory basis is $\{\sin(mNk_F),\cos(mNk_F)\}$, $m=1,2,3$, so that the $m=1$ quadratures carry the phase of the predicted term in Eq.~\eqref{eq:p4-residual}. Using $2\ell k_F$ instead rotates the fitted quadratures and does not reproduce the sine coefficients quoted below. The same convention is used for the charged fit and for the $p=2$ flux test. The $2k_F$ amplitudes are modeled as quadratic polynomials in $(\ell-\ell_0)/20$, and the $4k_F$ and $6k_F$ amplitudes as linear polynomials. Let $\widehat A_{\rm res}(\ell_0)$ be the fitted sine coefficient of the $2k_F$ component at the center $\ell_0$. Equation~\eqref{eq:p4-example} predicts the coefficient
\begin{equation}
 B_2=-\frac23\frac{\Gamma(7/8)}{\Gamma(1/8)}=-0.096422\ldots .
 \label{eq:p4-numerical-limit}
\end{equation}
With $\cL_0=4\ell_0\sin k_F$, we compare the scaled fitted estimator $\cL_0^{3/4}\widehat A_{\rm res}(\ell_0)$ with this analytic target.

\begin{figure}[H]
 \centering
 \IfFileExists{figures/fig_p4_uncharged_charged.pdf}{
 \includegraphics[width=\linewidth]{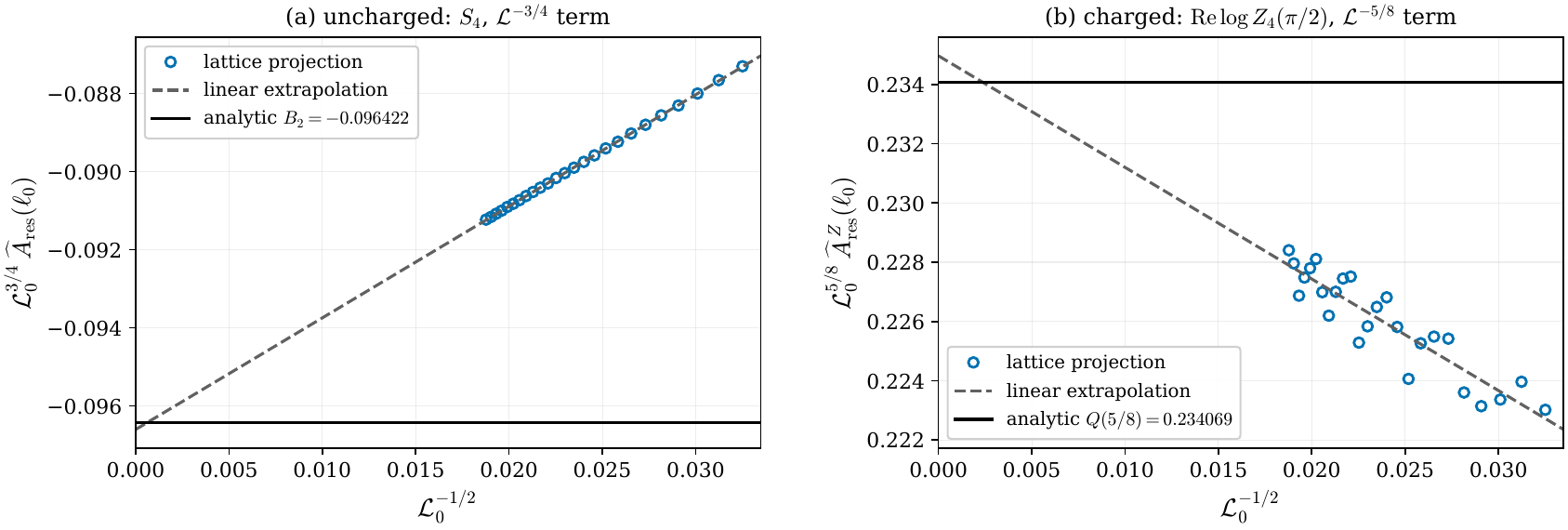}
 }{}
 \caption{Isolation of the subleading first-harmonic terms at $p=4$ and $k_F=1.4$. (a)~Uncharged entropy: open circles are obtained from the correlation matrix of the finite lattice block after subtracting the $(5/48)\log\ell$ term and the $\cL^{-1/4}$ prediction, as in Eq.~\eqref{eq:p4-residual}, and projecting locally onto $2k_F$. The ordinate is scaled so that the $\cL^{-3/4}$ coefficient tends to a constant, and the black line is the analytic value $B_2=-(2/3)\Gamma(7/8)/\Gamma(1/8)$ with no fitted normalization. (b)~Charged moment at $\phi=\pi/2$: the same protocol applied to $\Re\log Z_4(\pi/2)$ after subtracting the predicted $\cL^{-1/8}$, $\cL^{-3/8}$, and $\cL^{-7/8}$ terms. The ordinate is scaled so that the $\cL^{-5/8}$ coefficient tends to a constant, and the black line is the analytic value $Q(5/8)=\Gamma(13/16)/\Gamma(3/16)$. Dashed lines are empirical linear extrapolations in $\cL_0^{-1/2}$.}
 \label{fig:p4-subleading}
\end{figure}

The empirical linear intercept is $-0.096606$, within $0.20\%$ of the analytic value $B_2=-0.096422$. A quadratic extrapolation gives $-0.096378$, within $0.05\%$. Varying the Gaussian width from 15 to 40 sites and the window half-width from 40 to 80 sites changes the linear intercept by less than $2\times10^{-5}$. The spread under these protocol variations is therefore an order of magnitude smaller than the $0.20\%$ offset from $B_2$, consistent with residual curvature of the extrapolation in $\cL_0^{-1/2}$. Within this fitting basis, the quadratic extrapolation removes most of it. The numerical procedure is a test of the projected term, not an independent estimate of the remainder for the full entropy or of its convergence rate.

The charged test uses the same chain, filling, window, and fit protocol. For each $\ell$, $\log Z_4(\pi/2)$ is evaluated from the spectrum, with the principal logarithm of each factor. We subtract the three predicted contributions $\cL^{-1/8}$, $\cL^{-3/8}$, and $\cL^{-7/8}$ from the real part and project the residual locally onto $2k_F$, with fitted sine coefficient $\widehat A^{\,Z}_{\rm res}(\ell_0)$. Equation~\eqref{eq:p4-charged-example} predicts the coefficient
\begin{equation}
 Q(5/8)=\frac{\Gamma(13/16)}{\Gamma(3/16)}
 =0.234069\ldots .
 \label{eq:p4-charged-limit}
\end{equation}
We compare $\cL_0^{5/8}\widehat A^{\,Z}_{\rm res}(\ell_0)$ with this analytic target. Figure~\ref{fig:p4-subleading}(b) shows the comparison. The empirical linear intercept in $\cL_0^{-1/2}$ is $0.234971$, within $0.39\%$ of the analytic value. Including nuisance harmonics through the eighth changes this intercept to $0.235476$ with the same window, about $0.60\%$ above the analytic coefficient. The extrapolated intercept depends on the fitting basis.

Two further checks use no subtraction at all. The local sine component of $\Re\log Z_4(\pi/2)$ matches the predicted sum $\varSigma_-+\varSigma_+$, and the local cosine component of $\Im\log Z_4(\pi/2)$ matches the predicted difference $\varSigma_--\varSigma_+$, to better than $5\times10^{-4}$ in relative terms at every fit center. The two predicted contributions in Eq.~\eqref{eq:quadratures} are therefore separately visible in the lattice calculation. After subtracting all four predicted terms from both components, fits over finite windows give effective decay powers $1.19$ and $1.08$. They lie on opposite sides of $9/8$, the threshold in the rigorous $O(\cL^{-9/8})$ remainder bound. The visible preasymptotic drift prevents these fits from determining an asymptotic remainder exponent.

At $p=2$, we also test the flux enhancement discussed in Section~\ref{sec:charged}. For this test, the fitted $2k_F$ amplitude of $\Re\log Z_2(\phi)$ is the Euclidean norm of its sine and cosine coefficients. Between $\phi=0$ and $\phi=\pi/2$, this amplitude increases by factors $5.447$ at $k_F=1.0$ and $5.498$ at $k_F=\pi/3$. These values come from a least-squares fit over $\ell\in[340,460]$ to a constant, $\log\ell$, and the sine and cosine components at $2k_F$. Applying the same estimator to Eq.~\eqref{eq:main-charged} gives $5.449$ and $5.488$, with no fitted parameter. At $k_F=\pi/3$, exact aliasing of the $4k_F$ harmonic onto $2k_F$ can contribute to the residual difference.

\section{Conclusion}
\label{sec:conclusion}

For the open XX chain at fixed $\alpha>1$ and $k_F$ away from $0$ and $\pi$, the leading term in the extracted $2k_F$ component of the boundary R\'enyi entropy has the phase of the Friedel oscillation sampled at the entangling bond. The correlation determinant reduces exactly to a Hankel determinant with one internal jump. Riemann--Hilbert analysis and the exact determinant recurrence determine the two exponential terms that produce this oscillation.

For fixed $\alpha>1$, the calculation retains every first-harmonic residue on the initial sheet whose decay power is below one. For noninteger index, the additional branch-cut integral is controlled by a uniform positive-jump estimate and an exact spectral bound. At integer index, exact factorization gives a finite sum over all roots and an $O(\cL^{-1-1/p})$ remainder within the extracted component. This derives the corresponding part of the generalized Fisher--Hartwig expansion with the stated error. Boundary reflection fixes the phase, and the R\'enyi index controls the sequence of decay powers.

At fixed $\alpha>1$ and fixed flux with all $a_{q,\mp}\ne1$, the determinant theorem and the contour for $F_{\alpha,\phi}$ yield every first-harmonic residue on the initial sheet with $a_{q,\mp}<1$, with explicit gamma-function coefficients and a controlled remainder for the extracted component. A nonzero flux assigns unequal decay powers to the two exponential terms. At integer index, combining the known factorization over shifted fluxes with the boundary determinant expansion sharpens the remainder at each fixed $\phi\in(-\pi,\pi)$, including the interior crossings. Fixed-charge asymptotics would require uniform control of the complete $Z_\alpha(\phi)$ over the full Fourier interval, which is not established here. Here $[\rho_A,Q_A]=0$, so the entanglement asymmetry vanishes~\cite{AresMurcianoCalabrese2023}.

\paragraph{Acknowledgments.}
I thank Filiberto Ares, Pasquale Calabrese, Ricardo Esp\'indola Romero, Benoit Estienne, Leonardo Santilli, and Germ\'an Sierra for helpful correspondence and comments.

\par\bigskip
\appendix
\section{From a positive jump to a complex jump}
\label{app:complex-RH}

This appendix proves Theorem~\ref{thm:complex-jump}. Its purpose is to separate three ingredients that were combined in the main text: the result taken from the positive-jump analysis of Foulqui\'e Moreno, Mart\'inez-Finkelshtein, and Sousa, the extension needed when the jump is complex, and the passage from recurrence coefficients to the two determinant coefficients. The full steepest-descent construction is not repeated. Throughout, $0<\delta<\pi/2$ is fixed, $k_F\in[\delta,\pi-\delta]$, and $n\to\infty$ through positive integers.

\subsection{The published positive-jump input}

Ref.~\cite{FoulquieMorenoMartinezFinkelshteinSousaJAT2010} treats a generalized Jacobi weight with one positive jump at an arbitrary point $x_0\in(-1,1)$. The specialization used here is
\begin{equation}
 \begin{array}{c|c}
 \text{data in Ref.~\cite{FoulquieMorenoMartinezFinkelshteinSousaJAT2010}}
 & \text{value in the present problem} \\ \hline
 x_0 & \cos k_F \\
 \text{endpoint exponents} & \alpha_{\rm J}=\beta_{\rm J}=\tfrac12 \\
 \text{interior algebraic exponent} & \gamma=0 \\
 \text{analytic multiplier} & h(x)=1 \\
 \text{jump ratio} & c^2=e^{-2\pi i\beta} \\
 \text{jump parameter and branch} & c=e^{-\pi i\beta},\quad \log c=-\pi i\beta
 \end{array}
 \label{eq:FMMS-dictionary}
\end{equation}
The parameter called $\lambda$ in the confluent-hypergeometric model of that reference is therefore
$i\log c/\pi=\beta$. Thus $c>0$, the hypothesis in the published theorem, corresponds exactly to $\beta\in i\mathbb R$.

Theorem~1 of Ref.~\cite{FoulquieMorenoMartinezFinkelshteinSousaJAT2010} states the positive-jump asymptotics. The calculation leading to its Proposition~6, in particular Eqs.~(82)--(84), together with the recurrence reconstruction in its Eq.~(85), gives the asymptotic expansion of $a_n$ whose first correction is needed here. After the specialization in Eq.~\eqref{eq:FMMS-dictionary}, and before the two conjugate terms are combined into the cosine used in that theorem, that correction is
\begin{align}
 a_n(\beta)-\frac12
 &=-\frac{\beta^2\sin k_F}{4n}
 \left(\tau_\beta e^{2i\eta_n}
       +\tau_{-\beta}e^{-2i\eta_n}\right)+O_{k_F,\beta}(n^{-2}),
 \label{eq:FMMS-specialized-recurrence}\\
 \tau_\beta&=-\frac{\Gamma(-\beta)}{\Gamma(1+\beta)},
 \qquad
 \eta_n=nk_F+\frac{k_F}{2}-\frac{\pi}{4}
          -i\beta\log(4n\sin k_F),
 \qquad \beta\in i\mathbb R .
 \label{eq:FMMS-tau-eta}
\end{align}
This is an algebraic rewriting of the published positive-jump formula, not a new asymptotic result on the imaginary axis. The order-$n^{-1}$ endpoint residues vanish in this specialization because their coefficients are $4\alpha_{\rm J}^2-1=4\beta_{\rm J}^2-1=0$, as shown in Eq.~(76) of that reference.

Using $N_n=2n+1$, $\cL_n=4n\sin k_F$, and
\begin{equation}
 \begin{gathered}
 e^{2i\eta_n}=-i e^{iN_nk_F}\cL_n^{2\beta},
 \qquad
 e^{-2i\eta_n}=i e^{-iN_nk_F}\cL_n^{-2\beta},\\
 \beta^2\tau_\beta=\frac{\Gamma(1-\beta)}{\Gamma(\beta)},
 \qquad
 \beta^2\tau_{-\beta}=\frac{\Gamma(1+\beta)}{\Gamma(-\beta)},
 \end{gathered}
 \label{eq:FMMS-to-Psi}
\end{equation}
Eq.~\eqref{eq:FMMS-specialized-recurrence} is precisely
$a_n-\frac12=-\sin^2k_F(\Psi_{n,+}+\Psi_{n,-})+O(n^{-2})$ on $i\mathbb R$, with $\Psi_{n,\pm}$ defined in Eqs.~\eqref{eq:Psi-plus} and \eqref{eq:Psi-minus}.

\subsection{What changes for a complex jump}

The final cosine form in the positive-jump theorem contains absolute values, arguments of Gamma functions, and complex conjugates. It must not be continued as written. Instead, one returns to the explicit local model in Proposition~3, its large-coordinate expansion in Lemma~4, and the matching and residue formulas in Eqs.~(71)--(73) and (82)--(85) of Ref.~\cite{FoulquieMorenoMartinezFinkelshteinSousaJAT2010}. The moment conditions are then the bilinear conditions in Eq.~\eqref{eq:complex-orthogonality}. Positivity and complex conjugation are not part of the Riemann--Hilbert problem.

Fix the logarithm in Eq.~\eqref{eq:FMMS-dictionary}. For fixed local coordinate away from the model rays, the sectorwise confluent-hypergeometric solution is holomorphic in $\beta$ on
\begin{equation}
 \mathfrak S_\beta=\{\beta\in\mathbb C:|\Re\beta|<1/2\}.
 \label{eq:complex-strip}
\end{equation}
The apparent singularity at $\beta=0$ is removable. The jump identities and the unit-determinant identity are holomorphic identities in $\beta$. Since they hold on $i\mathbb R$, the identity theorem extends them to $\mathfrak S_\beta$. The large-coordinate expansion is locally uniform when $\beta$ ranges over a compact subset of this strip. In the matching coefficients, the conjugate quantities used for $\beta\in i\mathbb R$ are replaced by their holomorphic partners evaluated at $-\beta$. The first matching coefficient is
\begin{equation}
 A_1(\beta)=\beta^2
 \begin{pmatrix}-1&\tau_\beta\\-\tau_{-\beta}&1\end{pmatrix}.
 \label{eq:complex-first-matching}
\end{equation}
Both off-diagonal entries $\beta^2\tau_{\pm\beta}$ in Eq.~\eqref{eq:complex-first-matching} are holomorphic at the origin by Eq.~\eqref{eq:FMMS-to-Psi}. The same replacement in every higher matching coefficient gives a locally uniform, holomorphic expansion. A complex-parameter version of the same confluent-hypergeometric construction is also used in Ref.~\cite{CharlierGharakhloo2021}.

The scalar Szeg\H{o} function and the outer parametrix are given by the same explicit formulas with $\log c=-\pi i\beta$. Together with the analytic local prefactor, they are holomorphic in $\beta$ and invertible where appropriate on compact subsets of $\mathfrak S_\beta$. They also retain the required local behavior at $x_0$.

Exact nonvanishing already gives solvability for every degree. The asymptotic estimates still require a uniform small-norm argument, which does not follow from the identity theorem. Choose a symmetric compact rectangle $K_0\Subset\mathfrak S_\beta$ whose interior contains $K\cup\{0\}$. Choose a second symmetric compact rectangle $K_1\Subset\mathfrak S_\beta$ whose interior contains $K_0$, and put
\begin{equation}
 b=\sup_{\beta\in K_1}|\Re\beta|<\frac12,
 \qquad \rho=1-2b>0.
 \label{eq:complex-rho}
\end{equation}
On the matching circle around $x_0$, the two off-diagonal terms have absolute sizes $n^{-1-2\Re\beta}$ and $n^{-1+2\Re\beta}$. The diagonal terms are $O(n^{-1})$. The endpoint Bessel models are unchanged~\cite{KuijlaarsRH2004}. Their matching terms carry Laurent power zero, and the order-$n^{-1}$ residues at the endpoints vanish for the exponents in Eq.~\eqref{eq:FMMS-dictionary}. Away from the three matching circles, the remaining jumps are exponentially small. Hence the jump $J_R$ of the final error problem satisfies
\begin{equation}
 \|J_R-I\|_{L^2\cap L^\infty}
 =O_{K_1,\delta}(n^{-\rho}).
 \label{eq:complex-small-norm}
\end{equation}
The singular-integral equation is therefore invertible by a Neumann series for all sufficiently large $n$, uniformly on $K_1$. Its solution is holomorphic in $\beta$. Reversing the steepest-descent transformations recovers the original Riemann--Hilbert solution and its uniform asymptotic expansion. The bilinear polynomials exist, and $H_m(\beta)\ne0$, for every degree throughout the strip by the exact product in Section~\ref{sec:determinant-reduction}.

We next isolate the two errors required by the theorem. Attach $\xi$ and $\xi^{-1}$ to the elementary factors $e^{-iN_nk_F}$ and $e^{iN_nk_F}$. For any resulting Laurent series, $(\cdot)_+$ and $(\cdot)_-$ denote the coefficients of $\xi$ and $\xi^{-1}$, respectively. We also write $\nabla d_n=d_{n+1}-d_n$. The rest of this subsection proves only that the remainder required in Theorem~\ref{thm:complex-jump} is smaller by one power of $n$. After the explicit phase and power of $n$ are removed, the coefficients in the matching expansion and in its Neumann-series solution have the form
\begin{equation}
 \xi^r e^{-irN_nk_F}n^{-j-2r\beta}
 d_{j,r,n}(\beta;k_F),
 \qquad
 \nabla^m d_{j,r,n}=O_{K_0,\delta,j,r,m}(n^{-m}),
 \label{eq:complex-tagged-terms}
\end{equation}
for $j\ge1$, $r\in\mathbb Z$, and every fixed $m\ge0$. For each fixed Laurent power $r$, first extract the coefficient of $\xi^r$ and factor out $e^{-irN_nk_F}n^{-2r\beta}$. Expanding the local matching relation to arbitrary order and solving the small-norm equation recursively gives the remaining amplitude a locally uniform asymptotic expansion in integer powers of $n^{-1}$, with coefficients holomorphic in $\beta$. The exponentially small jumps obey the same finite-difference bounds. Taking the expansion beyond order $m$ and then applying $\nabla^m$ to the normalized amplitude gives the bounds in Eq.~\eqref{eq:complex-tagged-terms}. Uniform convergence on a fixed compact annulus about $|\xi|=1$ justifies the Laurent extraction.

The powers of $\xi$ add under multiplication. The matching coefficients at the internal jump have powers $0,1,-1$, while the endpoint contributions have power zero. Multiplication by the $n$-independent analytic prefactors may mix matrix entries but not Laurent powers. The only terms with $r=\pm1$ and base order $j=1$ are those from $A_1$. Every other term with either of these Laurent powers has $j\ge2$. Applying the residue calculation and Eqs.~(82)--(85) of Ref.~\cite{FoulquieMorenoMartinezFinkelshteinSousaJAT2010} within the preceding complex small-norm expansion, and choosing the local square root for which $a_n\to1/2$, therefore gives
\begin{align}
 \left(a_n-\frac12\right)_+
 &=-\sin^2k_F\,\Psi_{n,+}
   +e^{-iN_nk_F}n^{-2-2\beta}r_{n,+},
 \label{eq:complex-recurrence-plus}\\
 \left(a_n-\frac12\right)_-
 &=-\sin^2k_F\,\Psi_{n,-}
   +e^{iN_nk_F}n^{-2+2\beta}r_{n,-},
 \label{eq:complex-recurrence-minus}
\end{align}
where $r_{n,\pm}$ are locally holomorphic in $\beta$ and
$\nabla^m r_{n,\pm}=O_{K_0,\delta,m}(n^{-m})$. The first terms in Eqs.~\eqref{eq:complex-recurrence-plus}--\eqref{eq:complex-recurrence-minus} are the continued Foulqui\'e Moreno--Mart\'inez-Finkelshtein--Sousa coefficients. The separate remainder estimates, each smaller by one power of $n$, are the additional conclusion of the Laurent bookkeeping.

\subsection{From the recurrence coefficient to the determinant}

The exact identity in Eq.~\eqref{eq:second-difference-summary} gives
\begin{equation}
 \Delta^2\log\cD_n(\beta)=\log\bigl(4a_n(\beta)^2\bigr).
 \label{eq:complex-curvature}
\end{equation}
The determinants do not vanish anywhere in the strip. Choose their holomorphic logarithms by continuation from $\beta=0$, and choose $\log(4a_n^2)$ consistently with their exact ratio. Both sides of Eq.~\eqref{eq:complex-curvature} then vanish at $\beta=0$, so the determinant identity leaves no undetermined multiple of $2\pi i$.

It remains only to invert $\Delta^2$ within each tagged oscillatory coefficient. If $d_n$ satisfies the finite-difference bounds in Eq.~\eqref{eq:complex-tagged-terms}, then
\begin{equation}
 \Delta^2\!\left(e^{\mp iN_nk_F}n^{-j\mp2\beta}d_n\right)
 =e^{\mp iN_nk_F}n^{-j\mp2\beta}
 \left[-4\sin^2k_F\,d_n+n^{-1}\widetilde d_n\right],
 \label{eq:complex-difference-action}
\end{equation}
where $\widetilde d_n$ satisfies the same type of bounds. Since $\sin k_F\ge\sin\delta$, the leading multiplier is uniformly nonzero. Successive correction in powers of $n^{-1}$, followed by two discrete summations by parts for the remainder, constructs an inverse with the same power of $n$ and the same finite-difference bounds. The partial sums of $e^{\mp2ink_F}$ are uniformly bounded on the stated filling interval. The kernel of $\Delta^2$ consists of affine sequences. Neither of the $r=\pm1$, $j\ge1$ oscillatory decaying classes used here contains a nonzero affine sequence. The inverse is therefore unique within those classes.

Finally, products of the two leading $r=\pm1$ recurrence terms have Laurent powers $0$ or $\pm2$, not $\pm1$. Thus the leading $r=\pm1$ part of the right-hand side of Eq.~\eqref{eq:complex-curvature} is linear in Eqs.~\eqref{eq:complex-recurrence-plus} and \eqref{eq:complex-recurrence-minus}. The factor $4$ below follows from $\log(4a_n^2)=4(a_n-\frac12)+O((a_n-\frac12)^2)$. There are amplitudes $s_{n,\pm}$ and $t_{n,\pm}$ such that
\begin{align}
 \bigl[\log(4a_n^2)\bigr]_{\pm}
 &=-4\sin^2k_F\,\Psi_{n,\pm}
   +e^{\mp iN_nk_F}n^{-2\mp2\beta}s_{n,\pm},\notag\\
 \Delta^2\Psi_{n,\pm}
 &=-4\sin^2k_F\,\Psi_{n,\pm}
   +e^{\mp iN_nk_F}n^{-2\mp2\beta}t_{n,\pm},
 \label{eq:complex-final-cancellation}
\end{align}
where $s_{n,\pm}$ and $t_{n,\pm}$ are locally holomorphic in $\beta$, and $\nabla^m s_{n,\pm}$ and $\nabla^m t_{n,\pm}$ are $O_{K_0,\delta,m}(n^{-m})$. The identical nonzero leading factors in these two equations give Eqs.~\eqref{eq:det-sector-plus} and \eqref{eq:det-sector-minus}, including the bounds in Eq.~\eqref{eq:q-symbol-bounds}. This proves Theorem~\ref{thm:complex-jump}. The scalar normalization in Eq.~\eqref{eq:exact-physical-normalization} changes only the Laurent power zero and therefore does not affect either coefficient.

\subsection{Uniform control on a growing part of the imaginary axis}
\label{app:positive-jump-uniform}

The branch-cut integrals below sample $\beta=\pm it$ with unbounded $t\ge0$. Asymptotics uniform on compact parameter sets do not by themselves justify their estimation. We need the following extension only up to a logarithmically growing cutoff.

\begin{lemma}
\label{lem:imaginary-uniform}
Fix $A>0$ and $0<\delta<\pi/2$. For $k_F\in[\delta,\pi-\delta]$ and $\beta=\pm it$, $0\le t\le A\log n$, Eqs.~\eqref{eq:det-sector-plus}--\eqref{eq:det-sector-minus} hold with
\[
 |\nabla^m q_{n,\pm}(\beta;k_F)|
 \le C_{A,\delta,m}(1+t)^{M_m}n^{-m}
\]
for each fixed $m\ge0$ and some finite $M_m$. In particular, for some finite $M$,
\[
 \left|\bigl(\log\cD_n(\pm it)\bigr)_+
       \right|+
 \left|\bigl(\log\cD_n(\pm it)\bigr)_-
       \right|
 \le \frac{C_{A,\delta}}{n}(1+t)^M .
\]
\end{lemma}

\begin{proof}
Here the jump is positive. Write $\sigma_3=\operatorname{diag}(1,-1)$ and conjugate the error problem by $D_\infty^{-\sigma_3}$, where $D_\infty$ is the value at infinity of the scalar Szeg\H{o} function. We track the parameter in the local matching, the error equation, and the discrete inversion.

Choose the internal matching disk to have radius $r_0/(1+t)$ and the endpoint disks to have radii $r_0/(1+t)^2$, with $r_0>0$ small enough uniformly in $k_F$. After the conjugation, the internal analytic prefactor at $x_0$ is a bounded matrix times $e^{i\eta_n\sigma_3}$, with $\eta_n$ in Eq.~\eqref{eq:FMMS-tau-eta} real. Away from its value at $x_0$, its dependence on $t$ is through exponentials of $\beta$ times analytic functions of size $O(|z-x_0|)$. Thus the prefactor and its inverse are bounded on the shrinking circle. At an endpoint the corresponding vanishing functions are $O(|z\mp1|^{1/2})$. The endpoint radii keep their exponentials bounded as well. All remaining powers of the local coordinates have at most polynomial growth in $1+t$. These statements follow by canceling the weight factors against the scalar Szeg\H{o} function in the explicit prefactors of Ref.~\cite{FoulquieMorenoMartinezFinkelshteinSousaJAT2010}, Eqs.~(66)--(69), and in the endpoint construction of Ref.~\cite{KuijlaarsRH2004}.

We next bound the complete internal model, including its sector factors. Let $\Psi(\zeta;\beta)$ and $M(\zeta;\beta)$ be the matrices in Lemma~4 of Ref.~\cite{FoulquieMorenoMartinezFinkelshteinSousaJAT2010}, specialized to interior exponent zero, and put
\[
 V(\zeta;\beta)=\Psi(\zeta;\beta)M(\zeta;\beta)
      e^{\zeta\sigma_3/2}\zeta^{\beta\sigma_3}.
\]
The powers use that model's branch. This normalization removes the exponential and power factors before a matrix norm is taken. In the sectors adjacent to $\Re\zeta=0$, the entries are expressed directly through $U(a,1,\zeta)$ and its rotated arguments, with $a=\pm it$ or $1\pm it$. The factors $e^{\pm\pi i\beta}$ cancel against $M$ and the rotated powers. In the remaining sectors the connection formulas add a subdominant exponential only where $|\Re\zeta|\ge c|\zeta|$, for a fixed $c>0$. Its coefficient is bounded by $e^{C(1+t)}$.

For the algebraic remainders, Eq.~(13.7.5) of Ref.~\cite{DLMF} applies with $|(b-2a)/\zeta|\le1/4$ once $n$ is sufficiently large in the stated range. In Eqs.~(13.7.6)--(13.7.9) of that reference the auxiliary constants are then bounded, except for $\rho=O((1+t)^2)$. After removing $\zeta^{-a}$, the Pochhammer products give a bound $C_m(1+t)^{2m}|\zeta|^{-m}$ for a remainder truncated after $m$ terms. The Gamma ratios in the off-diagonal entries have magnitude $t$ or $t^{-1}$. The latter are multiplied by Pochhammer products vanishing quadratically at $t=0$. Consequently, uniformly on the closed matching sectors,
\[
 \begin{aligned}
 V(\zeta;\beta)
 &=I+\sum_{j=1}^{m-1}\frac{A_j(\beta)}{\zeta^j}
     +\mathcal R_m(\zeta;\beta),\qquad
 \|A_j(\pm it)\|\le C_j(1+t)^{2j},\\
 \|\mathcal R_m(\zeta;\pm it)\|
 &\le C_m(1+t)^{2m}|\zeta|^{-m}
       \exp\!\left(\frac{C_m(1+t)^2}{|\zeta|}\right)
       +C_m e^{C_m(1+t)-c|\zeta|}.
 \end{aligned}
\]
The coefficient $A_1$ is Eq.~\eqref{eq:complex-first-matching}. On the internal circle, $|\zeta|\asymp n/(1+t)$, so the algebraic remainder is $O(n^{-m}(1+t)^{3m})$ for $t\le A\log n$. The Bessel models at the endpoints have fixed orders. Their remainder bounds, combined with the prefactor bounds above, likewise give $n^{-m}$ times a polynomial in $1+t$. Their first residues vanish. Lens jumps and subdominant exponentials are bounded by $e^{C(1+t)-cn/(1+t)}$, which is smaller than every fixed power of $n^{-1}$ in this range. The prefactor $e^{C(1+t)}$ bounds the jump ratio and its inverse, together with the sector connection factors, for both signs of $\beta=\pm it$.

These estimates also hold for the Laurent extraction. In the internal matching matrix, insert $\xi^{-\sigma_3/2}$ beside $e^{i\eta_n\sigma_3}$ and its inverse on the other side. This multiplies the two elementary oscillatory factors by $\xi^{-1}$ and $\xi$, respectively. Only integer powers of $\xi$ occur in the resulting conjugation. On a fixed annulus $r\le|\xi|\le r^{-1}$, $0<r<1$, it changes the matrix bounds by a constant. Let $\widetilde J_R$ be the conjugated error jump, with these factors inserted. For some finite $d,e\ge0$,
\[
 \|\widetilde J_R-I\|_{L^2\cap L^\infty}
   \le \frac{C(1+t)^d}{n},
 \qquad \|C_-\|_{L^2\to L^2}\le C(1+t)^e.
\]
Here $C_-$ is the Cauchy boundary operator on the error contour. Choose the lens arcs so that, after rescaling each shrinking neighborhood, the contour consists of finitely many smooth arcs with uniformly bounded geometry and junction angles bounded away from tangency. The local Cauchy operators, including the junctions, are then uniformly bounded. The remaining disjoint pieces have separation at least $c(1+t)^{-2}$ and bounded length, so their interaction operators grow at most polynomially in $1+t$. The norm of the error integral operator is therefore at most $C(1+t)^{d+e}/n=o(1)$, uniformly in the stated range and in the $\xi$ annulus.

The Neumann series now converges uniformly. Expanding the matching to any fixed order and using the preceding bounds gives the normalized recurrence remainders and their finite differences polynomial bounds in $1+t$. For these differences $t$ is held fixed, so the matching circles do not move with $n$. Differences of $n^{\pm2it}$ add only polynomial factors in $t$. Laurent extraction is justified by Cauchy's formula on the fixed annulus. The first internal residue is unchanged by shrinking its circle, and hence gives the same $\Psi_{n,\pm}$.

Finally, invert Eq.~\eqref{eq:complex-difference-action} as in Appendix~\ref{app:complex-RH}. The multiplier is bounded away from zero, and two discrete summations by parts preserve polynomial dependence on $t$: differences of the normalized amplitudes have just been bounded, while the partial sums of $e^{\mp2ink_F}$ are bounded uniformly in the filling. This gives the stated bounds for $q_{n,\pm}$. Since $|\cL_n^{\pm2it}|=1$ and $|\Gamma(1+it)/\Gamma(-it)|=t$, each $\Psi_{n,\pm}$ is bounded by $C_\delta t/n$, proving the final assertion.
\end{proof}

\section{Entropy contour and the finite residue sum}
\label{app:entropy-extraction}

We prove Theorem~\ref{thm:real-alpha}. For noninteger $\alpha>1$, deformation on the initial sheet leaves both residues and a branch-cut integral. At integer index the result follows directly from Appendix~\ref{app:integer-factorization}.

Define
\begin{equation}
 \widehat D_\ell(\lambda)=\lambda^{-\ell}D_\ell(\lambda).
 \label{eq:Dhat}
\end{equation}
Since $\widehat D_\ell(\lambda)=\prod_{j=1}^{\ell}(1-x_j/\lambda)$, it is analytic and nonvanishing in the exterior of $[-1,1]$, with the unique logarithm normalized by $\log\widehat D_\ell(\infty)=0$. This logarithm extends to the complement of the convex hull of $\{0,x_1,\ldots,x_\ell\}$, so it is defined along the spectral contour $\mathcal C$. Equation~\eqref{eq:entropy-parts} gives
\begin{equation}
 \mathcal I_{\ell,\alpha}
 :=S_\alpha(\ell)-\ell\log2
 =-\frac{1}{2\pi i}\oint_{\mathcal C}
 e_\alpha'(\lambda)\log\widehat D_\ell(\lambda)\,d\lambda,
 \qquad
 \delta S_\alpha^{(1)}=[\mathcal I_{\ell,\alpha}]^{(1)}.
 \label{eq:exact-projected-integral}
\end{equation}
The contour $\mathcal C$ is chosen as in Eq.~\eqref{eq:entropy-contour}. The term $\ell\log2$ has Laurent power zero. We first deform this exact integral, before extracting its first harmonic.

For $t\ge0$, put
\begin{equation}
 w_\alpha(t)=\frac{2\alpha\sin(\pi\alpha)}{1-\alpha}
 \frac{e^{-2\pi\alpha t}}
 {1+2\cos(\pi\alpha)e^{-2\pi\alpha t}+e^{-4\pi\alpha t}}.
 \label{eq:branch-cut-weight}
\end{equation}
For noninteger $\alpha$, this function is bounded at zero and decays as $O_\alpha(e^{-2\pi\alpha t})$ at infinity. The parametrizations of the two real cuts satisfy
\begin{equation}
 \lambda=\sigma\coth(\pi t),\qquad
 \beta(\lambda)=-\sigma it,
 \qquad \sigma\in\{-1,1\},\quad t>0.
 \label{eq:beta-on-cuts}
\end{equation}
The exact decomposition is
\begin{equation}
 \begin{aligned}
 \mathcal I_{\ell,\alpha}
 &=\frac{1}{1-\alpha}
   \sum_{\substack{q\ge1\\2q-1<\alpha}}\sum_\pm
       \log\widehat D_\ell(\lambda_{q,\pm})
   +\mathcal B_{\ell,\alpha},\\
 \mathcal B_{\ell,\alpha}
 &=\int_0^\infty w_\alpha(t)
   \left[\log\widehat D_\ell(\coth\pi t)
        +\log\widehat D_\ell(-\coth\pi t)\right]dt .
 \end{aligned}
 \label{eq:residue-decomposition}
\end{equation}

To verify the cut term and its sign, write $u=(1+x)/2$ and $v=(x-1)/2$ for $x>1$. The upper and lower boundary values of $F_\alpha$ are $u^\alpha+e^{-i\pi\alpha}v^\alpha$ and $u^\alpha+e^{i\pi\alpha}v^\alpha$. Apply the residue theorem to the region between $\mathcal C$, a large circle, and both banks of the exterior cuts. The positive cut contributes
\[
 -\frac{1}{2\pi i}\int_1^\infty
   (e_{\alpha,+}'-e_{\alpha,-}')(x)
   \log\widehat D_\ell(x)\,dx.
\]
Substitution of $x=\coth(\pi t)$ and $v/u=e^{-2\pi t}$ gives Eq.~\eqref{eq:branch-cut-weight}. Reflection gives the negative cut. The small endpoint circles vanish because $\alpha>1$ and the determinant zeros lie strictly inside $(-1,1)$. At fixed $\ell$, $e_\alpha'=O(\lambda^{-1})$ and $\log\widehat D_\ell=O(\lambda^{-1})$, so the large-circle contribution tends to zero. The crossed poles are exactly the finite initial-sheet zeros in Eq.~\eqref{eq:kernel-zeros}. This proves Eq.~\eqref{eq:residue-decomposition} on the slit plane without a quotient of a branch cover.

We next estimate the cut integral. For either sign and all $t\ge0$, the exact product gives
\[
 \bigl|\log\widehat D_\ell(\sigma\coth\pi t)\bigr|
 \le 2\pi\ell t.
\]
Indeed, each factor is $1-\sigma x_j\tanh(\pi t)>0$, and its logarithm has absolute value at most $2\pi t$. The tail beyond $T=A\log\ell$ is therefore bounded by
$C_\alpha\ell(1+T)e^{-2\pi\alpha T}$. Given any fixed $B>0$, choose $A$ so this is $O(\ell^{-B})$, uniformly in the filling. This estimate is made on the exact integral before any Laurent extraction.

On $0\le t\le T$, Lemma~\ref{lem:imaginary-uniform} applies. The scalar normalization between $\widehat D_\ell$ and $\cD_\ell$ has Laurent power zero. Uniform convergence on the fixed $\xi$ annulus permits extraction under the truncated integral. Its two coefficients are bounded by
$C\ell^{-1}\int_0^T(1+t)^M e^{-2\pi\alpha t}\,dt$. The exact tail is smaller than any prescribed retained algebraic order, so it does not change the coefficients of this asymptotic expansion. At fillings with lattice aliasing the labels are inherited from the same expansion uniform in $k_F$. Thus
\begin{equation}
 [\mathcal B_{\ell,\alpha}]^{(1)}
 =O_{\alpha,\delta}(\cL^{-1}).
 \label{eq:branch-cut-bound}
\end{equation}

At $\lambda_{q,+}$, Eq.~\eqref{eq:positive-zero-residue} gives the retained exponential term, while Eq.~\eqref{eq:negative-zero-residue} gives its counterpart at $\lambda_{q,-}$. Their sum is
\begin{equation}
 2Q_{\alpha,q}\cL^{-a_{\alpha,q}}\sin(Nk_F).
 \label{eq:q-pair-sum}
\end{equation}
For each $q$, the error with the same exponential factor is $O(\cL^{-1-a_{\alpha,q}})$, and the leading term with the other factor is $O(\cL^{-2+a_{\alpha,q}})$. Since $0<a_{\alpha,q}<1$, both are smaller than $\cL^{-1}$ by a fixed power. Combining the finite residue sum with Eq.~\eqref{eq:branch-cut-bound} proves Eq.~\eqref{eq:main-real-alpha} for noninteger $\alpha$.

For integer $\alpha=p\ge2$, Appendix~\ref{app:integer-factorization} proves the stronger Corollary~\ref{cor:integer-p}, which implies Theorem~\ref{thm:real-alpha}. If $2q-1=\alpha$, the corresponding tangent in Eq.~\eqref{eq:kernel-zeros} is infinite, so it is not a finite residue. The strict inequality excludes it. The noninteger cut estimates are not asserted to be uniform as $\alpha$ approaches an odd integer.

\section{Exact factorization at integer index}
\label{app:integer-factorization}

We prove Corollary~\ref{cor:integer-p}. The equation $P_p(x)=0$ is equivalent to
\begin{equation}
 \left(\frac{1+x}{1-x}\right)^p=-1.
 \label{eq:Pp-root-equation}
\end{equation}
Let
\begin{equation}
 K_p=\left\{k\in\mathbb Z:-\frac p2<k+\frac12<\frac p2\right\},
 \qquad
 r_k=i\tan\frac{\pi(k+1/2)}{p}.
 \label{eq:Kp-rk}
\end{equation}
The finite roots are $r_k$ with $k\in K_p$. For even $p$, the polynomial has degree $p$ and leading coefficient $2^{1-p}$. For odd $p$, its degree is $p-1$ and its leading coefficient is $p2^{1-p}$. Thus
\begin{equation}
 P_p(x)=c_p\prod_{k\in K_p}(x-r_k),
 \label{eq:Pp-factorization}
\end{equation}
where $c_p$ is defined in Eq.~\eqref{eq:cp}. The degree is even in both cases.

Using $D_\ell(\lambda)=\prod_j(\lambda-x_j)$,
\begin{equation}
 \prod_{j=1}^{\ell}P_p(x_j)
 =c_p^\ell\prod_{k\in K_p}D_\ell(r_k),
 \label{eq:product-over-spectrum}
\end{equation}
with no additional sign. Pairing conjugate roots gives Eq.~\eqref{eq:integer-factorization-main} exactly.

For the root $r_{q,+}$ in the upper half-plane,
\begin{equation}
 \beta(r_{q,+})=-\frac12+\frac{2q-1}{2p},
 \qquad
 \beta(r_{q,-})=\frac12-\frac{2q-1}{2p}.
 \label{eq:beta-integer-roots}
\end{equation}
Theorem~\ref{thm:complex-jump} therefore supplies the retained pair
\begin{equation}
 2Q_{p,q}\cL^{-(2q-1)/p}\sin(Nk_F).
 \label{eq:integer-q-pair}
\end{equation}
After multiplication by $(1-p)^{-1}$ and summation over
$q=1,\ldots,\lfloor p/2\rfloor$, this gives the finite sum in Eq.~\eqref{eq:main-integer-p}.

It remains to identify the slowest omitted first-harmonic power. At each pair, the next term with the displayed exponential factor is $O(\cL^{-1-a_{p,q}})$, and the leading term with the other exponential factor is $O(\cL^{-2+a_{p,q}})$. For even $p$, the order $O(\cL^{-1-1/p})$ is attained by both the $q=1$ remainder with the same exponential factor and the other exponential factor of the last pair. For odd $p$, it is attained by the $q=1$ remainder with the same exponential factor. All remaining omitted first-harmonic terms decay faster. The number of roots is finite, which proves the error $O_{p,\delta}(\cL^{-1-1/p})$.

Quadratic products of the leading oscillatory terms may occur at order $\cL^{-2/p}$, but they have zero frequency or frequency $4k_F$. They are excluded from $\delta S_p^{(1)}$ and do not alter the extracted first-harmonic error.

\section{Charged moments: contour and factorization}
\label{app:charged}

This appendix proves Eq.~\eqref{eq:fcs-identity}, Theorem~\ref{thm:charged}, Eq.~\eqref{eq:twist-factorization}, and Corollary~\ref{cor:charged-integer}.

\subsection{The counting determinant}

Fix $\phi\in(-\pi,\pi)\setminus\{0\}$. For one correlation eigenvalue $x\in[-1,1]$, put $\nu=(1+x)/2\in[0,1]$.
\begin{equation}
 1-\nu+\nu e^{i\phi}
 =\frac{1-e^{i\phi}}{2}\,(\lambda_*-x),
 \qquad
 \lambda_*=-\frac{e^{i\phi}+1}{e^{i\phi}-1}=i\cot\frac{\phi}{2}.
 \label{eq:fcs-factor}
\end{equation}
The product over the spectrum of $2C-I_\ell$ gives Eq.~\eqref{eq:fcs-identity}. Direct computation gives
\begin{equation}
 \lambda_*+1=\frac{-2}{e^{i\phi}-1},
 \qquad
 \lambda_*-1=\frac{-2e^{i\phi}}{e^{i\phi}-1},
 \qquad
 \frac{\lambda_*+1}{\lambda_*-1}=e^{-i\phi},
 \label{eq:fcs-ratio}
\end{equation}
so $\beta(\lambda_*)=-\frac{\phi}{2\pi}$ on the principal branch for $|\phi|<\pi$. The scalar prefactor of Eq.~\eqref{eq:fcs-identity} is affine in $\ell$ after taking logarithms and is independent of $\xi$. Corollary~\ref{cor:fcs} follows from Theorem~\ref{thm:complex-jump} at $\beta=-\frac{\phi}{2\pi}$ together with the specializations
\begin{equation}
 \frac{\Gamma(1+\beta)}{\Gamma(-\beta)}\bigg|_{\beta=-\frac{\phi}{2\pi}}
 =Q(1-\phi/\pi),
 \qquad
 \frac{\Gamma(1-\beta)}{\Gamma(\beta)}\bigg|_{\beta=-\frac{\phi}{2\pi}}
 =Q(1+\phi/\pi).
 \label{eq:fcs-Q}
\end{equation}
At $\phi=0$, the statement follows directly from $Z_1(0)=1$ and the local holomorphy of Theorem~\ref{thm:complex-jump} at $\beta=0$. This also establishes the asserted uniformity for compact flux sets that contain zero.

\subsection{Charged factor and contour}

We prove Theorem~\ref{thm:charged}. The zeros of $F_{\alpha,\phi}$ solve $((1+\lambda)/(1-\lambda))^{\alpha}=-e^{-i\phi}$. On the initial sheet fixed below Eq.~\eqref{eq:renyi-kernel},
\begin{equation}
 \frac{1+\lambda}{1-\lambda}
 =\exp\!\left[\frac{i(2q-1)\pi-i\phi}{\alpha}\right],
 \qquad q\in\mathbb Z,
 \qquad
 \left|\frac{(2q-1)\pi-\phi}{\alpha}\right|<\pi.
 \label{eq:charged-root-equation}
\end{equation}
For $q\ge1$ this gives $\mu_{q,-}$ in Eq.~\eqref{eq:charged-zeros}, and the substitution $q\to1-q$ gives the family $\mu_{q,+}$. The sheet condition is exactly $a_{q,\mp}<1$. The zeros are simple, since $F_{\alpha,\phi}=F_{\alpha,\phi}'=0$ would force $(1+\lambda)/2=-(1-\lambda)/2$, which has no solution. Hence $F_{\alpha,\phi}'/F_{\alpha,\phi}$ has residue one at each zero. For $\lambda=i\tan\psi$ with $\psi\in(0,\pi/2)$, one has $(\lambda+1)/(\lambda-1)=-e^{2i\psi}$ and $\beta=\psi/\pi-1/2$. For $\lambda=-i\tan\psi$, the same computation gives $(\lambda+1)/(\lambda-1)=-e^{-2i\psi}$ and $\beta=1/2-\psi/\pi$. The choices $\psi=\pi a_{q,-}/2$ and $\psi=\pi a_{q,+}/2$ give Eq.~\eqref{eq:charged-beta}. At fixed $\phi\ne0$ the two families are not complex conjugates of each other. Conjugation maps the zeros of $F_{\alpha,\phi}$ to those of $F_{\alpha,-\phi}$, consistent with $\overline{Z_\alpha(\phi)}=Z_\alpha(-\phi)$.

For noninteger $\alpha$, use the same normalized exterior logarithm as in Appendix~\ref{app:entropy-extraction}. Subtracting its pole at zero before integration by parts gives the affine term $\ell\log F_{\alpha,\phi}(0)$, with the branch continued from $\phi=0$. Deformation on the initial sheet then gives the exact identity
\begin{equation}
 \log Z_\alpha(\phi)-\ell\log F_{\alpha,\phi}(0)
 =\sum_{\substack{q\ge1,\ s\in\{-,+\}\\a_{q,s}<1}}
    \log\widehat D_\ell(\mu_{q,s})
   +\mathcal B_{\ell,\alpha,\phi},
 \label{eq:charged-decomposition}
\end{equation}
where
\begin{equation}
 \begin{aligned}
 \mathcal B_{\ell,\alpha,\phi}
 &=\sum_{\sigma=\pm1}\int_0^\infty
 W_{\alpha,\phi,\sigma}(t)
 \log\widehat D_\ell(\sigma\coth\pi t)\,dt,\\
 W_{\alpha,\phi,\sigma}(t)
 &=\frac{2\alpha\sin(\pi\alpha)e^{-i\sigma\phi}e^{-2\pi\alpha t}}
 {1+2e^{-i\sigma\phi}\cos(\pi\alpha)e^{-2\pi\alpha t}
     +e^{-2i\sigma\phi}e^{-4\pi\alpha t}}.
 \end{aligned}
 \label{eq:charged-cut-weight}
\end{equation}
Here $\sigma$ labels the positive or negative real cut, independently of the labels of the roots. On the positive cut the boundary values of the kernel are
$e^{i\phi}u^\alpha+e^{\mp i\pi\alpha}v^\alpha$, with $u=(1+x)/2$ and $v=(x-1)/2$. Their logarithmic derivatives give $W_{\alpha,\phi,+}$ by the same oriented calculation as in Appendix~\ref{app:entropy-extraction}. Reflection replaces $\phi$ by $-\phi$ and gives $W_{\alpha,\phi,-}$.

For the fixed admissible fluxes in Theorem~\ref{thm:charged}, the denominators are nonzero at $t=0$. For $t>0$ this follows by factoring them as
$[1+e^{-i\sigma\phi+i\pi\alpha}e^{-2\pi\alpha t}]
 [1+e^{-i\sigma\phi-i\pi\alpha}e^{-2\pi\alpha t}]$.
Thus $W_{\alpha,\phi,\sigma}$ is bounded at zero and is $O_{\alpha,\phi}(e^{-2\pi\alpha t})$ at infinity. The exact spectral tail bound and Lemma~\ref{lem:imaginary-uniform} give
$[\mathcal B_{\ell,\alpha,\phi}]^{(1)}=O_{\alpha,\phi,\delta}(\cL^{-1})$
by precisely the truncation and extraction argument of Appendix~\ref{app:entropy-extraction}. The affine normalization has Laurent power zero.

At the zero $\mu_{q,-}$, Theorem~\ref{thm:complex-jump} gives $ie^{-iNk_F}Q(a_{q,-})\cL^{-a_{q,-}}$ for the coefficient proportional to $e^{-iNk_F}$. Its error is $O(\cL^{-1-a_{q,-}})$, and the term proportional to $e^{iNk_F}$ at the same zero is $O(\cL^{-2+a_{q,-}})$. The family $\mu_{q,+}$ gives $-ie^{iNk_F}Q(a_{q,+})\cL^{-a_{q,+}}$ with the analogous errors. Every retained decay power satisfies $0<a_{q,\mp}<1$, so all errors are $o(\cL^{-1})$ by a fixed power. Summing over the finitely many retained zeros and adding the cut bound proves Theorem~\ref{thm:charged} for noninteger $\alpha$. For integer $\alpha=p$, the factorization below proves Corollary~\ref{cor:charged-integer}, which implies Theorem~\ref{thm:charged} at its admissible fluxes.

\subsection{Factorization over shifted fluxes and the sharpened remainder}

In this subsection, use the decay powers $a_{q,\mp}^{[p]}$ for integer $p$ defined in Corollary~\ref{cor:charged-integer}. Let $z_r=e^{i\phi_r}$ with $\phi_r$ as in Eq.~\eqref{eq:twist-angles}. Then $z_r^{\,p}=(-1)^{p+1}e^{i\phi}$, so the numbers $-z_r$ are precisely the $p$ roots of $w^p=-e^{i\phi}$. For $0<\nu<1$, with the endpoint cases obtained by continuity,
\begin{equation}
 \prod_{r=0}^{p-1}\left[(1-\nu)+\nu z_r\right]
 =\nu^p\prod_{r=0}^{p-1}\left(\frac{1-\nu}{\nu}-(-z_r)\right)
 =\nu^p\left[\left(\frac{1-\nu}{\nu}\right)^{\!p}+e^{i\phi}\right]
 =(1-\nu)^p+\nu^pe^{i\phi}.
 \label{eq:twist-root-identity}
\end{equation}
The product over the correlation eigenvalues gives Eq.~\eqref{eq:twist-factorization}, which is the product over the discrete fluxes in Eq.~(2.37) of Ref.~\cite{CapizziMurcianoCalabrese2023}. References~\cite{KlichLevitov2009,SongBipartite2012} give broader earlier relations between entanglement and full counting statistics. We define $\log Z_p(\phi)$ as the sum of the $p$ factor logarithms. For $\phi_r\ne0$, Eq.~\eqref{eq:fcs-identity} fixes each factor as $\ell\operatorname{Log}[(1-e^{i\phi_r})/2]$ plus the exterior determinant logarithm continued from $\lambda=\infty$ along $\lambda_*(t\phi_r)$, $0<t\le1$. The value at $\phi_r=0$ is its continuous limit. Thus no branch ambiguity enters the projection. Since $|\phi_r|\le(|\phi|+(p-1)\pi)/p<\pi$, every factor satisfies Corollary~\ref{cor:fcs} uniformly.

The two exponential terms in the factor $r$ have decay powers $1\mp\phi_r/\pi$. The reindexings $q=p-r$ and $q=r+1$ give
\begin{equation}
 1-\frac{\phi_r}{\pi}=a_{p-r,-}^{[p]},
 \qquad
 1+\frac{\phi_r}{\pi}=a_{r+1,+}^{[p]},
 \label{eq:twist-enumeration}
\end{equation}
so the sum over shifted fluxes enumerates $a_{q,\mp}^{[p]}$ for $q=1,\ldots,p$. The remainders in the coefficients of $\xi$ and $\xi^{-1}$ are $O(\cL^{-2+\phi_r/\pi})$ and $O(\cL^{-2-\phi_r/\pi})$. The slowest over $r$ is $O(\cL^{-1-\kappa_p(\phi)})$ with $\kappa_p(\phi)=(1-|\phi|/\pi)/p$. Enumerated terms with $a_{q,\mp}^{[p]}\ge1+\kappa_p(\phi)$ are absorbed into that remainder. This proves Corollary~\ref{cor:charged-integer}. Quadratic products of the leading $\xi$ and $\xi^{-1}$ terms have powers $0$ or $\pm2$ and are excluded from the extracted first harmonic. Higher nonlinear terms proportional to $\xi$ or $\xi^{-1}$ are included in the corresponding remainder, as in Appendix~\ref{app:integer-factorization}.

\paragraph*{Use of AI tools.}
The author developed the main idea and first results of this work years before using generative AI tools.
During the preparation and revision of this manuscript, the author used ChatGPT and Claude to assist with the development and checking of mathematical arguments and calculations, the preparation of code for computational checks, and the revision of the exposition.
The author takes full responsibility for the contents of the manuscript.

\end{document}